%% file: main.tex
\documentclass[a4paper,UKenglish,cleveref, autoref, thm-restate]{lipics-v2021}
\pdfoutput=1 %uncomment to ensure pdflatex processing (mandatatory e.g. to submit to arXiv)
\hideLIPIcs  %uncomment to remove references to LIPIcs series (logo, DOI, ...), e.g. when preparing a pre-final version to be uploaded to arXiv or another public repository

\title{Dynamic Grammar-Compressed Self-Index in \texorpdfstring{$\delta$}{delta}-Optimal Space} %TODO Please add

\author{Takaaki Nishimoto}{RIKEN Center for Advanced Intelligence Project, Japan}{takaaki.nishimoto@riken.jp}{https://orcid.org/0009-0008-3798-8397}{}%mandatory, please use full name; only 1 author per \author macro; first two parameters are mandatory, other parameters can be empty.
\author{Yasuo Tabei}{RIKEN Center for Advanced Intelligence Project, Japan}{yasuo.tabei@riken.jp}{https://orcid.org/0000-0003-2368-5607}{}
\authorrunning{T. Nishimoto and Y. Tabei}%mandatory. First: Use abbreviated first/middle names. Second (only in severe cases): Use first author plus 'et al.'
\Copyright{Takaaki Nishimoto and Yasuo Tabei}%mandatory, please use full first names. LIPIcs license is "CC-BY";  

\ccsdesc[500]{Theory of computation~Pattern matching} %TODO mandatory: Please choose ACM 2012 classifications from https://dl.acm.org/ccs/ccs_flat.cfm 

\keywords{Compressed string indexes, grammar-compression, dynamic data structures} %TODO mandatory; please add comma-separated list of keywords

\supplementdetails[linktext={https://github.com/TNishimoto/dynamic-rlslp-index}]{Software}{https://github.com/TNishimoto/dynamic-rlslp-index} %linktext, cite, and subcategory are optional

\funding{This work was supported by JST, ACT-X Grant Number JPMJAX24CL, Japan.}%optional, to capture a funding statement, which applies to all authors. Please enter author specific funding statements as fifth argument of the \author macro.

\nolinenumbers %uncomment to disable line numbering

\EventEditors{Philip Bille, Seth Pettie, and Sabine Storandt}
\EventNoEds{3}
\EventLongTitle{34th Annual European Symposium on Algorithms (ESA 2026)}
\EventShortTitle{ESA 2026}
\EventAcronym{ESA}
\EventYear{2026}
\EventDate{August 31--September 4, 2026}
\EventLocation{L'Aquila, Italy}
\EventLogo{}
\SeriesVolume{388}
\ArticleNo{4}
\begin{document}
\input{tex/macros}

\maketitle

\input{tex/abst}
\input{tex/1_intro}

\input{tex/2_preliminaries}
\input{tex/3_restricted_recompression}

\input{tex/4_0_dag}
\input{tex/4_1_improvement}

\input{tex/7_update}
\input{tex/8_experiment}

\bibliography{ref}
\clearpage
\appendix
\input{tex/9a_appendix}

\input{tex/9b_appendix}

\input{tex/9c_appendix}

\input{tex/9d_appendix}

\end{document}

%% file: tex/macros.tex
%\newcommand{\floor}[1]{\left \lfloor #1 \right \rfloor}
%\newcommand{\ceil}[1]{\left \lceil #1 \right \rceil}

%%Other
%\newcommand{\argmax}{\mathop{\rm arg~max}\limits}
%\newcommand{\argmin}{\mathop{\rm arg~min}\limits}
%\newcommand{\polylog}{\mathop{\rm polylog}\limits}

%% Basic Notation
\newcommand{\reverse}{\mathsf{reverse}}
\newcommand{\Occ}{\mathit{Occ}}
\newcommand{\pOcc}{\mathit{pOcc}}
\newcommand{\occ}{\mathsf{occ}}
\newcommand{\avg}{\mathsf{avg}}

%%% Range Report

%%% Grammar
\newcommand{\val}{\mathsf{val}}
\newcommand{\expr}{\mathsf{expr}}
\newcommand{\assign}{\mathsf{assign}}

%%%% DAG
\newcommand{\DAG}{\mathsf{DAG}}
\newcommand{\dagraph}{\mathfrak{D}} % DAG symbol (distinct from \mathcal{D} = production rules)
\newcommand{\expl}{\mathsf{expl}}
\newcommand{\impl}{\mathsf{impl}}
\newcommand{\pathset}{\Pi} % set of paths in the path decomposition (formerly \pathset)

%%%% Point
\newcommand{\grid}{\mathsf{grid}}
\newcommand{\rect}{\mathsf{rect}}

%%%% Locate
\newcommand{\rocc}{\mathit{rocc}}
\newcommand{\vOcc}{\mathsf{vOcc}}

%%%% Update
\newcommand{\diff}{\mathsf{diff}}
\newcommand{\change}{\mathsf{change}}
\newcommand{\children}{\mathsf{children}}
%\newcommand{\tree}{\mathsf{tree}}

%%%% Appendix
\newcommand{\short}{\mathsf{short}}
\newcommand{\leaf}{\mathsf{leaf}}

%%%% Experiment
%\newcommand{\SEI}{\mathsf{SEI}}
%\newcommand{\RRI}{\mathsf{RRI}}

%%%% DEBUG
%\newcommand{\unproven}{\color{red}UNPROVEN\color{black}}
%\newcommand{\DESCRIPTION}{\color{red}DESCRIPTION\color{black}}
%\newcommand{\DUMMY}{\color{red}DUMMY DUMMY \color{black}}

%%%%% 

%\newcommand{\rmv}{\mathsf{rmv}}

%% file: tex/abst.tex
\begin{abstract}
A compressed self-index stores a string in compressed form while supporting locate queries without decompression.
For highly repetitive strings---arising in web crawls, versioned documents, and genomic collections---static self-indexes can match the $\delta$-optimal lower bound of $\Omega(\delta \log(n \log \sigma / (\delta \log n)) \log n)$ bits up to constant factors, where $n$ is the string length, $\sigma$ is the alphabet size, and $\delta$ is the substring complexity.
Their dynamic counterparts, however, remain scarce: every existing dynamic self-index either fails to attain $\delta$-optimal space, pays $\Omega(\log n)$ time per reported occurrence for locate queries, or has an update time that grows with the maximum value in the longest common prefix (LCP) array of the text.
We present the \emph{dynamic RR-index}, a dynamic grammar-compressed self-index built on the restricted recompression run-length straight-line program (RLSLP).
To our knowledge, it is the first dynamic self-index to attain $\delta$-optimal space.
The index uses $\mathcal{O}(\delta \log(n \log \sigma / (\delta \log n)) \log n)$ bits in expectation, answers locate queries in expected $\mathcal{O}(m + \log m \log^{2} n + \occ (\log n / \log \log n))$ time---where $m$ is the pattern length and $\occ$ is the number of occurrences---and supports insertion of a length-$m'$ string and deletion of a length-$m'$ substring in expected amortized $\mathcal{O}(m' \log^{2} n + \log^{3} n)$ time, with no dependence on the maximum LCP value.
On eleven highly repetitive corpora, including a $37$~GB Wikipedia dump and a $59$~GB human-chromosome collection, the dynamic RR-index is up to $77\times$ faster than the dynamic r-index for updates and up to $11\times$ faster than other dynamic indexes for locate queries.
\end{abstract}

%% file: tex/1_intro.tex
\section{Introduction}\label{sec:intro}

Datasets composed of \emph{highly repetitive} strings---those with many repeated substrings---have grown rapidly in recent years.
Typical examples include web pages collected by crawlers~\cite{FerraginaM10}, version-controlled documents such as Wikipedia with its complete edit history~\cite{dataset:enwiki-all-pages}, and, perhaps most notably, biological sequences such as collections of human genomes~\cite{Przeworski00,1000Genomes}.
These datasets are not static: edits, deletions, and new releases arrive continuously.
Yet the methods practitioners rely on most heavily, such as the r-index~\cite{DBLP:journals/jacm/GagieNP20}, a compressed index supporting locate queries, are \emph{static}: reflecting even a single edit requires rebuilding the index from scratch.
There is therefore a strong and growing need to develop \emph{dynamic} compressed self-indexes for highly repetitive strings that can be updated efficiently without rebuilding.

A practical solution must simultaneously (i) use space proportional to the repetitiveness of the string, (ii) answer locate queries quickly, and (iii) support string insertions and deletions without rebuilding.
The third requirement has proven difficult: while compressed \emph{self-indexes}---data structures that represent the string in compressed form while answering locate queries directly---are well understood in the static case, their dynamic counterparts remain challenging to develop.

For highly repetitive strings, compressed self-indexes aim to approach the \emph{$\delta$-optimal space} of $\mathcal{O}(\delta \log(n \log \sigma / (\delta \log n)) \log n)$ bits established by Kociumaka et al.~\cite{9961143}, where $n$ is the string length, $\sigma$ is the alphabet size, and $\delta$ is the \emph{substring complexity} of the string---a repetitiveness measure satisfying $\delta \ll n$ on genomic and versioned corpora.
A line of specialized self-indexes---grammar-based~\cite{DBLP:journals/talg/ChristiansenEKN21,9961143,ViceVersa}, LZ-based~\cite{DBLP:conf/latin/ChristiansenE18}, block-tree--based~\cite{DBLP:journals/jcss/BelazzouguiCGGK21,DBLP:journals/corr/abs-2602-13735,9961143}, 
and run-length BWT (RLBWT)-based indexes such as 
the r-index~\cite{DBLP:journals/jacm/GagieNP20} and OptBWTR~\cite{DBLP:conf/icalp/NishimotoT21,DBLP:conf/wea/Bertram0N24}---has closed much of this gap, with several attaining $\delta$-optimal space up to constant factors.
None, however, supports updates.

Only a handful of \emph{dynamic} self-indexes have been proposed, and each sacrifices at least one of the three requirements (Table~\ref{tab:dat} in Appendix~\ref{app:summarize_table} gives the full comparison).
The dynamic FM-index~\cite{DBLP:journals/jda/SalsonLLM10,DBLP:journals/jda/LeonardMS12} builds on the Burrows--Wheeler Transform (BWT)~\cite{burrows1994block} but uses $\mathcal{O}(n \log \sigma)$ bits, oblivious to repetitiveness, with an $\mathcal{O}(\log^{2+\epsilon} n)$ per-occurrence locate factor.
Grammar-based dynamic indexes---the SE-index~\cite{DBLP:journals/dam/NishimotoIIBT20}, the TST-index-d~\cite{DBLP:journals/iandc/NishimotoTT20}, and the index of Gawrychowski et al.~\cite{DBLP:journals/corr/GawrychowskiKKL15}---either pay $(\log n \log^{*} n)^{2}$-type factors in locate/update or require $\Omega(n \log n)$ bits.
A recent dynamic r-index~\cite{dynamic_r_index_dcc} uses $\mathcal{O}(r \log n)$ bits ($r$ the number of BWT runs), but its update time $\mathcal{O}((m' + L_{\max}) \log n)$ depends on the maximum value $L_{\max} \leq n$ of the longest common prefix (LCP) array---the length of the longest repeated substring---which can grow independently of $\delta$ and $r$ as versioned copies accumulate.
In short, every existing dynamic self-index either fails to compress to $\delta$-optimal space, pays $\Omega(\log n)$ time per reported occurrence for locate queries, or has an update time that grows with the maximum LCP value. We address this gap.

\subparagraph*{Our contributions.}
We present the \emph{dynamic RR-index}, a dynamic grammar-compressed self-index built on the \emph{restricted recompression} RLSLP of Je\.{z}~\cite{DBLP:journals/talg/Jez15} and Kociumaka et al.~\cite{DBLP:journals/corr/KociumakaRRW13,9961143}.
Our starting observation is that the derivation tree of a restricted recompression RLSLP admits a compact \emph{DAG} representation whose explicit-node count is tied to $\delta$ rather than to $r$ or $g$.
Moreover, insertions and deletions perturb this DAG only along two short
popped sequences~\cite{DBLP:journals/dam/NishimotoIIBT20,DBLP:conf/cpm/I17},
independently of the maximum LCP value.
Combined with two-dimensional range reporting~\cite{DBLP:conf/soda/Blelloch08,NAVARRO20142} and with a path-cached ancestor field that accelerates locate (\cref{sec:improvement}), this design simultaneously achieves $\delta$-optimal space and a sub-logarithmic per-occurrence locate bound.
To our knowledge, the dynamic RR-index is the first dynamic self-index to attain $\delta$-optimal space.
Concretely, it uses $\delta$-optimal space---$\mathcal{O}(\delta \log(n \log \sigma / (\delta \log n)) \log n)$ bits in expectation---and answers locate queries in expected $\mathcal{O}(m + \log m \log^{2} n + \occ (\log n / \log \log n))$ time, where $m$ is the pattern length and $\occ$ is the number of occurrences.
Insertion of a length-$m'$ string and deletion of a length-$m'$ substring are supported in expected amortized $\mathcal{O}(m' \log^{2} n + \log^{3} n)$ time.
We implement and evaluate the dynamic RR-index on eleven highly repetitive corpora, including the $37$~GB \textsf{enwiki} Wikipedia dump and the $59$~GB \textsf{chr19} genomic collection, and we show that it is up to $77\times$ faster than the dynamic r-index on updates and up to $11\times$ faster than other dynamic indexes for locate queries, while using working memory that empirically scales with $\delta$.

\subparagraph*{Organization.}
Sections~\ref{sec:preliminary}--\ref{sec:recompression} fix notation and review the structural properties of restricted recompression. Sections~\ref{sec:component}--\ref{sec:locate_query} introduce the dynamic RR-index, its DAG representation, the locate algorithm, and the ancestor-path caching of \cref{sec:improvement}. Section~\ref{sec:update} gives the insertion and deletion algorithms, and Section~\ref{sec:experiment} reports the experiments. 
Omitted proofs and implementation details are deferred to Appendices~\ref{app:summarize_table}-\ref{app:experiment}.

%% file: tex/2_preliminaries.tex
\section{Preliminaries}\label{sec:preliminary}

\subparagraph{Basic notation.}
Let $T$ be a string of length $|T| = n \geq 2$ over a totally ordered alphabet $\Sigma$ of size $\sigma = n^{\mathcal{O}(1)}$. 
We denote the empty string by $\varepsilon$, the $i$-th character of $T$ by $T[i]$, 
the substring of $T$ from position $i$ to position $j$ by $T[i..j]$, 
and the reverse of $T$ by $\reverse(T) = T[n]T[n-1] \cdots T[1]$.
A prefix (resp.\ suffix) of $T$ is a substring beginning at position $1$ (resp.\ ending at position $n$); 
for a string $P \in \Sigma^{+}$, 
$T \cdot P$ denotes concatenation. We use the standard lexicographic order on $\Sigma^{*}$. 
The set of occurrence positions of $P$ in $T$ is $\Occ(T, P) = \{ i \in \{1, 2, \ldots, n-|P|+1\} \mid T[i..i+|P|-1] = P \}$. For a nonempty set $\mathcal{S}$ of integers, $\min \mathcal{S}$ and $\max \mathcal{S}$ denote its minimum and maximum. For a path $\mathbb{P}$ in a graph, let $|\mathbb{P}|$ denote the number of edges on $\mathbb{P}$.

\subparagraph{Model of computation.}
We use base-2 logarithms throughout this paper unless otherwise indicated. 
Our computation model is a unit-cost word RAM~\cite{DBLP:conf/stacs/Hagerup98} with multiplication, randomization, and machine word size $B = \Theta(\log n)$ bits. 
%We measure space in machine words; multiplying by $B$ gives the space in bits. 
We analyze our randomized algorithms against an \emph{oblivious adversary} (e.g., \cite{DBLP:conf/stoc/BeimelKMNSS22}): the sequence of updates to $T$ is chosen in advance, without access to the random bits used by our algorithms.

\subparagraph{Substring complexity.}
Substring complexity~\cite{DBLP:journals/algorithmica/RaskhodnikovaRRS13,DBLP:journals/talg/ChristiansenEKN21} measures the repetitiveness of a string.
For each $d \geq 1$, let $\mathcal{S}_{d} = \{ T[i..i+d-1] \mid 1 \leq i \leq n - d + 1 \}$ be the set of length-$d$ substrings of $T$.
The \emph{substring complexity} of $T$ is $\delta = \max_{1 \leq d \leq n} |\mathcal{S}_{d}| / d$.

The \emph{$\delta$-optimal space} is $\Theta(\delta \log \tfrac{n \log \sigma}{\delta \log n} \log n)$ bits: every string can be encoded within this space bound, and this bound is tight in $n$, $\sigma$, and $\delta$~\cite{9961143}. 

%\subsection{Grammar-Based Compression}\label{subsec:slp}
\subparagraph{Grammar-based compression.}
A grammar-based compression of $T$ is a context-free grammar $\mathcal{G} = (\mathcal{V}, \Sigma, \mathcal{D}, E)$ that generates exactly $T$, where $\mathcal{V} = \{X_1, X_2, \ldots, X_g\}$ is a set of \emph{nonterminals}, $\Sigma$ is the alphabet of $T$, and $\mathcal{D}$ is a set of \emph{production rules} $X_i \rightarrow \expr_i$ with $\expr_i \in (\mathcal{V} \cup \Sigma)^*$. 
A rule $X_i \rightarrow \expr_i$ expands $X_i$ to the sequence $\expr_i$.
The \emph{start symbol} $E \in \mathcal{V}$ does not appear on any right-hand side, while every other nonterminal appears on at least one right-hand side.
%The \emph{start symbol} $E \in \mathcal{V}$ does not appear on any right-hand side, while every other nonterminal does appear on at least one.
To avoid cycles, we require $j < i$ whenever $X_j$ appears in $\expr_i$. 
In addition, $\expr_{i} \neq \expr_{j}$ for any two distinct integers $i, j \in \{ 1, 2, \ldots, g \}$.

Each nonterminal $X \in \mathcal{V}$ derives a substring of $T$; in particular, the start symbol $E$ derives $T$ itself.
The tree whose root is labeled $E$ and whose leaves spell out $T$ is the \emph{derivation tree} of $\mathcal{G}$.
Let $\val(X)$ denote the substring of $T$ derived by $X \in \mathcal{V}$, 
and let $\val(X_{1}, X_{2}, \ldots, X_{d}) = \val(X_{1}) \cdot \val(X_{2}) \cdot \cdots \cdot \val(X_{d})$ 
for nonterminals $X_{1}, X_{2}, \ldots, X_{d} \in \mathcal{V}$. 

A \emph{straight-line program} (SLP)~\cite{DBLP:journals/njc/KarpinskiRS97} is a grammar-based compression in which each nonterminal produces either a character or a pair of nonterminals.
A \emph{run-length SLP (RLSLP)}~\cite{DBLP:conf/mfcs/TanimuraNBIT17} is an SLP whose nonterminals can additionally produce a single nonterminal or a repetition of a single nonterminal.
Specifically, $\mathcal{D} \subseteq \mathcal{V} \times (\Sigma \cup \mathcal{V}^+)$, and each production rule in $\mathcal{D}$ takes one of the following forms:
(i) $X_{i} \rightarrow c$ for $c \in \Sigma$;
(ii) $X_{i} \rightarrow X_j X_k$ ($j, k < i$) for distinct nonterminals $X_j, X_k \in \mathcal{V}$;
(iii) $X_{i} \rightarrow X_j$ ($j < i$) for $X_j \in \mathcal{V}$;
(iv) $X_{i} \rightarrow (X_{j})^{d}$ ($j < i$), where $(X_{j})^{d}$ denotes $d \geq 2$ successive copies of $X_j \in \mathcal{V}$.
Several algorithms build an RLSLP for a given string, e.g., \emph{signature encoding}~\cite{DBLP:conf/mfcs/TanimuraNBIT17}, \emph{recompression}~\cite{DBLP:journals/talg/Jez15}, \emph{signature grammar}~\cite{DBLP:conf/latin/ChristiansenE18}, and \emph{restricted block compression}~\cite{ViceVersa}. 

%% file: tex/3_restricted_recompression.tex
\section{Restricted Recompression}\label{sec:recompression}
\begin{figure}[t]
\makebox[\linewidth][c]{%
        \includegraphics[scale=0.7]{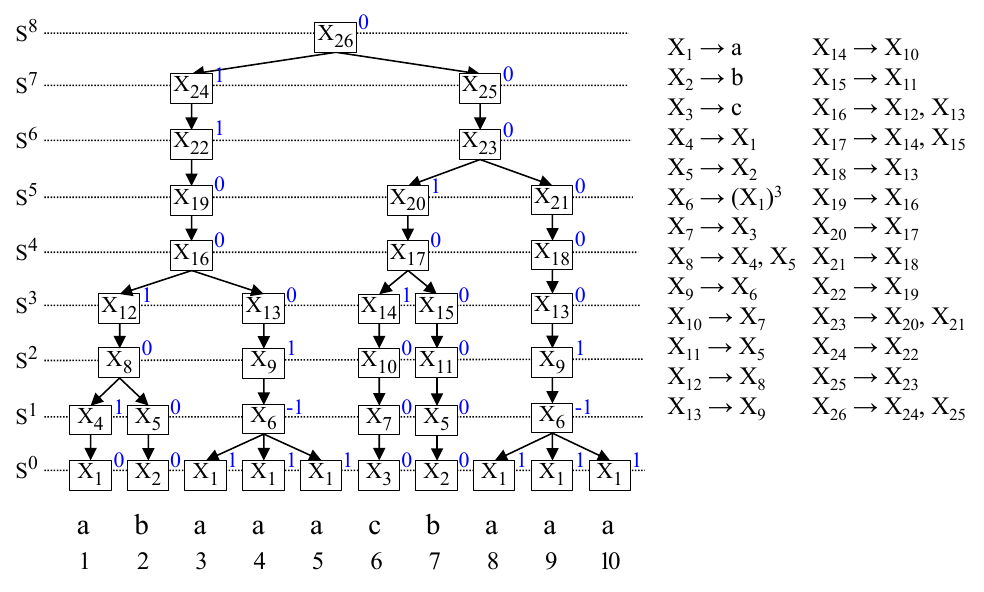}
}
	  \caption{
	  An illustration of the derivation tree of a restricted recompression RLSLP (left), and the corresponding production rules (right).
        Each node is depicted as a rectangle enclosing its corresponding nonterminal $X_{i}$.
        The blue integer shown at the upper right of each node is the value of $\assign(X_{i})$.
	  }
\label{fig:restricted_recompression}
\end{figure}

Restricted recompression~\cite{DBLP:journals/talg/Jez15,9961143,DBLP:journals/corr/KociumakaRRW13} is a randomized algorithm that constructs an RLSLP $\mathcal{G}^{R}$ whose derivation tree is \emph{height-balanced} (i.e., all children of every node have the same height).
Let $H$ denote the height of the derivation tree of $\mathcal{G}^{R}$.
For each $h \in \{0, 1, \ldots, H\}$, the nonterminals labeling the nodes at height $h$, read left to right, form a sequence $S^{h}$.
In this paper, $\mathcal{G}^{R}$ is called a \emph{restricted recompression RLSLP}. 

Restricted recompression builds the derivation tree bottom-up, processing $S^{0}, S^{1}, \ldots, S^{H}$ in order.
While processing $S^h$, each newly introduced nonterminal $X_{i}$ in $S^{h}$ is assigned a value $\assign(X_{i}) \in \{-1, 0, 1\}$ as follows.
Let $\mu(h) = (8/7)^{\lceil (h+1)/2 \rceil - 1}$.
If $|\val(X_{i})| \leq \mu(h)$, then $\assign(X_{i})$ is drawn uniformly at random from $\{0, 1\}$; otherwise, $\assign(X_{i}) = -1$.
The $\assign$ values partition $S^h$ into segments as follows.
For odd $h$,
every two consecutive nonterminals $X_{i}, X_{j}$ in $S^{h}$ with $\assign(X_i)=1$ and $\assign(X_j)=0$ form a segment;
for even $h$,
every maximal run of a nonterminal $X_{i}$ with $\assign(X_i) \in \{0,1\}$ forms a segment.
In either case, every remaining nonterminal forms a segment by itself.
Unless $|S^{h}| = 1$,
$S^{h+1}$ is constructed by replacing each segment with a nonterminal.
See Appendix~\ref{app:restricted_recompression_construction} for a detailed description of the construction.

Figure~\ref{fig:restricted_recompression} shows an example of the derivation tree for $T = \mathrm{abaaacbaaa}$.
For readability, Figure~\ref{fig:restricted_recompression} uses illustrative thresholds
$\tilde{\mu}(0)=\tilde{\mu}(1)=1$ and $\tilde{\mu}(h)=10$ for $2 \le h \le 8$,
instead of the actual thresholds $\mu(h)$.
%Since the actual values of $\mu(h)$ would make the tree too large to display, we use $\lfloor \mu(0) \rfloor = \lfloor \mu(1) \rfloor = 1$ and $\lfloor \mu(2) \rfloor = \lfloor \mu(3) \rfloor = \cdots = \lfloor \mu(8) \rfloor = 10$ in the figure.

The following lemma shows that $H = \mathcal{O}(\log n)$ with high probability.
\begin{lemma}\label{lem:tree_height}
$H \leq \lceil 2(w+1)\log_{8/7}(4n)+2 \rceil$ holds with probability at least $1-1/n^{w}$ for every integer $w \geq 1$.
\end{lemma}
\begin{proof}
See Appendix~\ref{app:tree_height}.
\end{proof}
For a user-defined constant $w \geq 2$, 
we assume $H \leq \lceil 2(w+1)\log_{8/7}(4n)+2 \rceil$ throughout this paper. 
By Lemma~\ref{lem:tree_height}, 
this bound holds with high probability.

Let $M$ denote the number of nonterminals in $\mathcal{V}$ excluding those that produce a single nonterminal.
By Proposition V.19 of~\cite{9961143}, $\mathbb{E}[M] = \mathcal{O}(\delta \log \tfrac{n \log \sigma}{\delta \log n})$. 
Using this bound on $M$, Section~\ref{sec:component} represents the restricted recompression RLSLP as a DAG that achieves $\delta$-optimal space in expectation.

%% file: tex/4_0_dag.tex
\section{Dynamic RR-index}\label{sec:component}
\begin{figure}[t]
\makebox[\linewidth][c]{%
\includegraphics[scale=0.55]{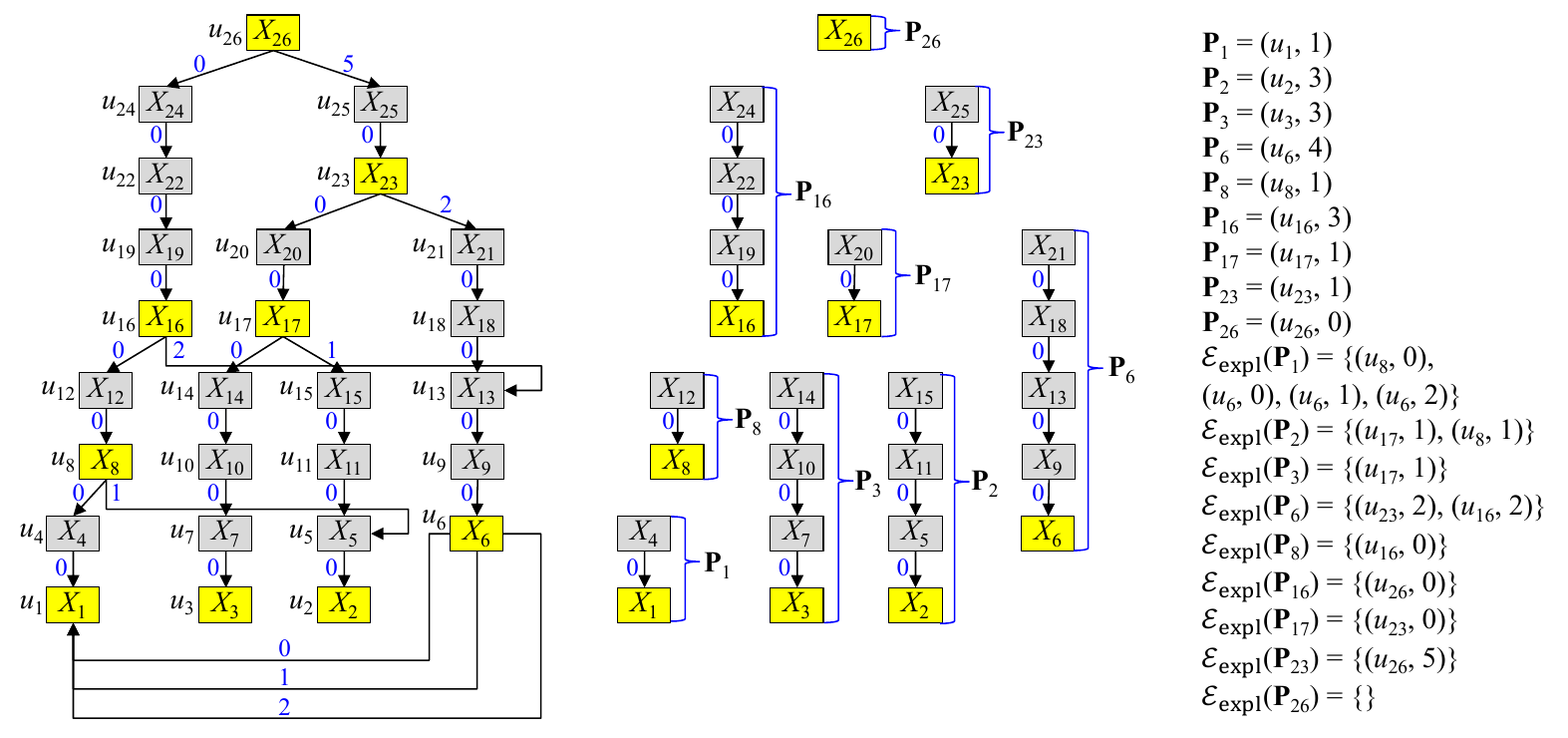}
}
\caption{(Left) DAG for the derivation tree of Figure~\ref{fig:restricted_recompression}: each rectangle is a node $u_{i}$ with nonterminal $X_{i}$ (yellow = explicit, gray = implicit), and the blue number above each edge is its label. (Center) The nine paths obtained by removing edges whose source is explicit. (Right) Each path together with its associated set of inter-path edges. For simplicity, each edge is represented as a pair consisting of its source and edge label.}
\label{fig:rrdag}
\end{figure}
The \emph{dynamic RR-index} is built on the restricted recompression RLSLP $\mathcal{G}^{R} = (\mathcal{V}, \Sigma, \mathcal{D}, E)$ introduced in the previous section.
It consists of two components: (i) a directed acyclic graph (DAG) representing the derivation tree, and (ii) a set of points in two-dimensional space corresponding to explicit nodes in the DAG. The details of these two components are described in the following subsections.

\subsection{DAG Representation of the Derivation Tree and Path Decomposition}

\subparagraph{DAG representation.}
The derivation tree of $\mathcal{G}^{R}$ is represented as a directed acyclic multigraph~\cite{DBLP:conf/mfcs/TanimuraNBIT17,DBLP:journals/dam/NishimotoIIBT20} $\dagraph = (\mathcal{U}, \mathcal{E}, \mathrm{src}, \mathrm{dst}, \mathcal{L}_U, \mathcal{L}_E)$ obtained by merging nodes labeled with the same nonterminal: $\mathcal{L}_U : \mathcal{U} \to \mathcal{V}$ is a bijection between nodes and nonterminals, $\mathcal{E}$ is the set of directed edges; for $e \in \mathcal{E}$, $\mathrm{src}(e)$ is its source (corresponding to a parent node in the derivation tree) and $\mathrm{dst}(e)$ its destination (a child node), and the outgoing edges of $u_i$ are in one-to-one correspondence with those of any node labeled $\mathcal{L}_U(u_i)$ in the derivation tree. The label $\mathcal{L}_E(e) \in \{0,1,\ldots,n-1\}$ is $j-i$, where $i$ and $j$ are the starting positions of the substrings derived from $\mathrm{src}(e)$ and $\mathrm{dst}(e)$. 
%Let $d^+(u)$ denote the outdegree of $u$. 
Let $d^+(u)$ denote the outdegree of $u$, counting parallel edges with multiplicity.
Figure~\ref{fig:rrdag} (left) illustrates this for Figure~\ref{fig:restricted_recompression}.

\subparagraph{Path decomposition.}
We partition $\mathcal{U}$ into
$\mathcal{U}_{\expl} = \{ u \in \mathcal{U} \mid d^+(u) \neq 1 \}
\qquad\text{and}\qquad
\mathcal{U}_{\impl} = \mathcal{U} \setminus \mathcal{U}_{\expl}$.
Nodes in $\mathcal{U}_{\expl}$ are called \emph{explicit}, and nodes in $\mathcal{U}_{\impl}$ are called \emph{implicit}.
After removing every edge whose source is explicit, the resulting subgraph $\dagraph'$ consists of $M$ disjoint directed paths, each ending at an explicit node (see Figure~\ref{fig:rrdag}, center). 

For each explicit node $u \in \mathcal{U}_{\expl}$, let $\mathbb{P}_u$ denote the directed path in $\dagraph'$ that ends at $u$.
Writing $\langle u, k \rangle$ for the implicit node $k$ edges before $u$ on this path, $\mathbb{P}_u$ has the form
\[
\mathbb{P}_u : \langle u,|\mathbb{P}_u| \rangle \rightarrow \langle u,|\mathbb{P}_u|-1\rangle \rightarrow \cdots \rightarrow \langle u,1\rangle \rightarrow u.
\]
Therefore, $\mathbb{P}_{u}$ can be encoded as the pair $(u, |\mathbb{P}_u|)$.

%where $\ell_u \ge 0$ is the number of edges in $\mathbb{P}_u$ (if $\ell_u = 0$, the path consists of the single node $u$).
The \emph{path decomposition} of $\dagraph$ is the set $\pathset = \{ \mathbb{P}_u \mid u \in \mathcal{U}_{\expl} \}$, so $|\pathset| = |\mathcal{U}_{\expl}| = M$ (see Section~\ref{sec:recompression}).
The directed edges of $\dagraph$ are partitioned into two kinds: \emph{intra-path edges}, which are the edges of $\dagraph'$ (equivalently, edges between consecutive nodes on the same path in $\pathset$), and \emph{inter-path edges}, whose source and destination lie on different paths in $\pathset$.

\subsection{Point Representation of Explicit Nodes}\label{subsubsec:range_report}
\begin{figure}[t]
 \begin{center}
		\includegraphics[scale=0.6]{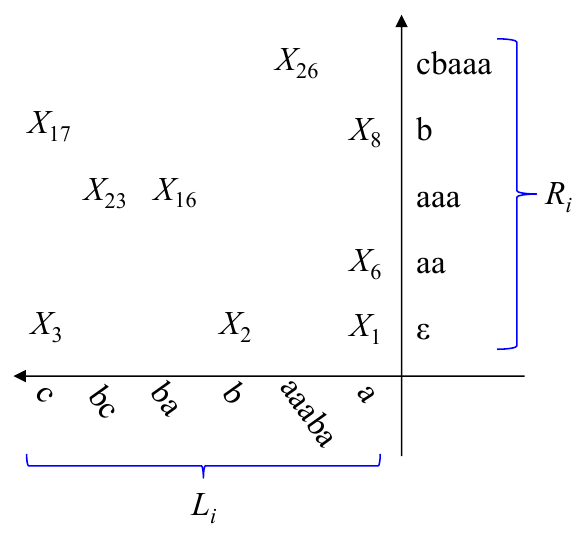}
	  \caption{Set $\mathcal{P}$: each nonterminal corresponds to a point in two-dimensional space.}
\label{fig:grid}
 \end{center}
\end{figure}

To support efficient locate queries, we map each explicit node $u_i \in \mathcal{U}_{\expl}$ to a point $p_i$ in two-dimensional space with $x$-coordinate $L_i$ and $y$-coordinate $R_i$, defined from the production rule of $X_i = \mathcal{L}_U(u_i)$ as follows: if $X_i \rightarrow X_j X_k$ then $L_i = \reverse(\val(X_{j}))$ and $R_i = \val(X_{k})$; if $X_i \rightarrow c$ for $c \in \Sigma$ then $L_i = c$ and $R_i = \varepsilon$; if $X_i \rightarrow (X_j)^{d}$ with $d \geq 2$ then $L_i = \reverse(\val(X_{j}))$ and $R_i = \val((X_j)^{d-1})$.
Let $\mathcal{P}$ be the resulting point set, and let $\mathcal{X}$ and $\mathcal{Y}$ be $\mathcal{U}_{\expl}$ sorted lexicographically by $L_i$ and $R_i$, with ties broken by $u_i$. 
This point representation will later allow us to find the relevant explicit nodes
for locate queries by two-dimensional range reporting, i.e., by queries that return all
points in a given axis-aligned rectangle. 
Figure~\ref{fig:grid} illustrates $\mathcal{P}$, $\mathcal{X}$, and $\mathcal{Y}$ for the DAG in Figure~\ref{fig:rrdag}.
%Primary occurrences of a pattern can then be found by \emph{range reporting queries}, i.e., queries returning all points in a given axis-aligned rectangle. 

\subsection{Dynamic Data Structures for the DAG and Point Set}\label{subsec:data_structures}
The DAG $\dagraph$ is indexed by $\mathscr{D}_{\DAG}(\pathset)$ and the point set $\mathcal{P}$ by $\mathscr{D}_{\grid}(\mathcal{P})$.

\subparagraph{Data structure for path decomposition.}
$\mathscr{D}_{\DAG}(\pathset)$ consists of a doubly linked list and two types of hash tables~\cite{DBLP:journals/siamcomp/DietzfelbingerKMHRT94}.
Intra-path edges are implicit in the list structure; inter-path edges are stored in per-path hash tables.
For each explicit node $u \in \mathcal{U}_{\expl}$ with $X_i := \mathcal{L}_U(u)$, 
the list stores one element associated with the endpoint $u$ that records: 
a pair $(u, |\mathbb{P}_u|)$ for recovering $\mathbb{P}_u$, 
the production rule $X_i \rightarrow \expr_i$, the height of $u$, $|\val(X_i)|$, and 
$\assign(\mathcal{L}_{U}(u_{j}))$ for each node $u_{j}$ on $\mathbb{P}_{u}$ 
(occupying $\mathcal{O}(H)$ bits since $|\mathbb{P}_u| \leq H+1$ and each label is $2$ bits).
For each path $\mathbb{P}_u \in \pathset$, 
the set $\mathcal{E}_{\expl}(\mathbb{P}_u) = \{ e \in \mathcal{E} \mid \mathrm{src}(e) \in \mathcal{U}_{\expl},\; \mathrm{dst}(e) \text{ lies on } \mathbb{P}_u \}$ is indexed by a per-path hash table whose entry for $e$ has key $\mathrm{src}(e)$ and value $\mathrm{dst}(e)$ (only distinct key-value pairs are stored). 
For a fixed path, an explicit source has at most one distinct child on that path; hence $\mathrm{src}(e)$ is sufficient as the key, and $\mathrm{dst}(e)$ is stored as the value. 
The keys are also linked in a doubly linked list sorted by the height of the corresponding nodes and indexed by a \emph{list indexing data structure}~\cite{DBLP:conf/wads/Dietz89}. 
It requires $\mathcal{O}(kB)$ additional bits and supports random access queries on the doubly linked list in $\mathcal{O}(\log k)$ time, where $k$ is the number of elements in the list; each insertion or deletion of an element takes amortized $\mathcal{O}(\log k)$ time.

A global hash table indexes the $M$ production rules with key $\expr_i$ (encoded in $\mathcal{O}(B)$ bits) and value $X_i$.

\subparagraph{Data structure for point set.}
We index $\mathcal{P} \subseteq \mathcal{X} \times \mathcal{Y}$ by Blelloch's dynamic data structure $\mathscr{D}_{\grid}(\mathcal{P})$~\cite{DBLP:conf/soda/Blelloch08} 
for range reporting and point insertion/deletion, which
uses $\mathcal{O}(M B)$ bits of space and requires $\mathcal{O}(1)$-time comparison of elements of $\mathcal{X}$ and $\mathcal{Y}$. 
We maintain \emph{order maintenance data structures}~\cite{DBLP:conf/stoc/DietzS87} on $\mathcal{X}$ and $\mathcal{Y}$ to provide $\mathcal{O}(1)$-time comparison and $\mathcal{O}(\log M)$-time insertion/deletion, and list indexing data structures for $\mathcal{O}(\log M)$-time access; each coordinate is a pointer to the corresponding list element of $\pathset$, so each point uses $\mathcal{O}(B)$ bits. Each list element of $\mathscr{D}_{\DAG}(\pathset)$ in turn stores pointers into the list indexing and order maintenance structures for $\mathcal{X}$ and $\mathcal{Y}$.
Range reporting takes $\mathcal{O}(\log M + \rocc(\log M / \log \log M))$, where $\rocc$ is the number of reported points, and a point insertion or deletion takes amortized $\mathcal{O}(\log M)$ time. 

\subsection{Space Complexity}
The doubly linked list contributes $\mathcal{O}(M(H+B))$ bits; the global hash table, the per-path hash tables (which hold $\mathcal{O}(M)$ distinct key-value pairs in total since every DAG node has at most two distinct children) together with their list indexing structures, and $\mathscr{D}_{\grid}(\mathcal{P})$ with its auxiliary structures each contribute $\mathcal{O}(MB)$ bits. Summing, the dynamic RR-index occupies $\mathcal{O}(M(H+B))$ bits.
Since $B = \Theta(\log n)$ and $\mathbb{E}[M] = \mathcal{O}(\delta \log \frac{n \log \sigma}{\delta \log n})$, the total space is $\delta$-optimal in expectation under the assumption that $H = \mathcal{O}(\log n)$.

\section{Locate Query Algorithm}\label{sec:locate_query}
This section presents the algorithm answering a locate query for a pattern $P$ of length $m$, establishing the baseline complexity bound (Theorem~\ref{thm:locate}); Section~\ref{sec:improvement} then improves the per-occurrence factor to $\log n / \log \log n$ via ancestor-path caching (Theorem~\ref{thm:locate_improved}).
For $m \geq 2$, the algorithm proceeds in four steps, each detailed in a subsection below: from the derivation tree we extract a specific \emph{core}~\cite{DBLP:journals/jda/MaruyamaNKS13} of $P$ as the \emph{popped sequence}~\cite{DBLP:journals/mst/DuysterK26,DBLP:conf/cpm/I17}, whose run-length encoding has $\rho$ runs, build $\rho$ axis-aligned rectangles in $\mathcal{P}$, apply range reporting to obtain the \emph{primary occurrences}~\cite{DBLP:journals/dam/NishimotoIIBT20}, and convert these into $\Occ(T, P)$ by traversing the DAG to the root. The case $m = 1$ is handled separately in Section~\ref{sec:locate_impl}.

\subsection{Step (i): Computing the Popped Sequence}
This step computes the popped sequence of $P$, a specific core of $P$ whose run-length encoding has $\mathcal{O}(\min\{m, H\})$ runs.
We first define a core and then present the popped sequence as a specific construction.

Let $Q = X_1, X_2, \ldots, X_d$ be a sequence of nonterminals (each $X_i \in \mathcal{V}$) such that $\val(Q) = P$.
We say that $Q$ is \emph{embedded} in the derivation tree of $\mathcal{G}^{R}$ at position $s \in \{1, 2, \ldots, n\}$ if, for every $i \in \{1, 2, \ldots, d\}$, some node labeled $X_i$ in the derivation tree derives the substring $T[s_{i}..s_{i} + |\val(X_{i})| - 1]$ of $T$, where $s_i := s + \sum_{b=1}^{i-1} |\val(X_b)|$.
Such an embedding witnesses an occurrence of $P$ at position $s$ (i.e., $s \in \Occ(T, P)$).
We call $Q$ a core of $P$ if $Q$ is embedded at every position $s \in \Occ(T, P)$.

Since $P$ can have multiple cores, we use a specific one, called the \emph{popped sequence} of $P$, defined as follows.
Restricted recompression is applied to $P$ itself, where production rules are shared with $\mathcal{G}^{R}$ whenever the right-hand sides agree.
At each height, the leftmost and rightmost segments of the current sequence are \emph{popped} (i.e., removed and recorded) unless they consist of two distinct nonterminals, and the other segments are replaced by their parent nonterminals as in Section~\ref{sec:recompression}.
The popped sequence $Q$ of $P$ is the concatenation of the popped left segments in bottom-up order, followed by the popped right segments in top-down order; $\val(Q) = P$ holds, and $Q$ is a core of $P$.
Its run-length encoding consists of $\rho$ runs, where $\rho = \mathcal{O}(\min\{m, H\})$ and $\mathbb{E}[\rho] = \mathcal{O}(\log m)$.
%Since a core of $P$ is not unique, we use a specific one called the popped sequence of $P$, whose $\rho$ runs satisfy $\rho = \mathcal{O}(\min\{m, H\})$ and $\mathbb{E}[\rho] = \mathcal{O}(\log m)$.
The popped sequence is computed from the derivation tree of $\mathcal{G}^R$ by the algorithm of~\cite{DBLP:journals/mst/DuysterK26}, invoked in phase (i) of Section~\ref{subsubsec:locate2}. 
See Appendix~\ref{app:popped_sequence} for the precise definition and properties of the popped sequence. 
The next step uses these $\rho$ runs to construct $\rho$ axis-aligned rectangles; Lemma~\ref{lem:core_rectangle} ensures this suffices for finding all primary occurrences.

\subsection{Step (ii): Computing the Rectangles from the Popped Sequence}

\subparagraph{Primary occurrences and their rectangles.}
A primary occurrence~\cite{DBLP:journals/dam/NishimotoIIBT20} of $P$ for a split position $\ell \in \{1,2,\ldots,m-1\}$ is a pair $(u_i, \ell)$ with $u_i \in \mathcal{U}_{\expl}$ such that (i) $\reverse(P[1..\ell])$ is a prefix of $L_i$ and (ii) $P[\ell+1..m]$ is a prefix of $R_i$; that is, $u_i$ is an explicit node where an occurrence of $P$ crosses the first boundary between substrings derived from $u_i$'s production rule.
We write $\pOcc(P,\ell) \subseteq \mathcal{U}_{\expl}$ for the set of such $u_i$ at split position $\ell$.
Recall from Section~\ref{subsubsec:range_report} that $u_{i}$ is mapped to the point $p_{i}$ with $x$-coordinate $L_{i}$ and $y$-coordinate $R_{i}$. The two prefix conditions constrain $L_i$ and $R_i$, respectively, so $\pOcc(P, \ell)$ coincides with the set of points of $\mathcal{P}$ in the axis-aligned rectangle $\rect_{\ell} = [x, x'] \times [y, y']$, where $x, x' \in \mathcal{X}$ are the smallest and largest elements whose $x$-coordinate has $\reverse(P[1..\ell])$ as a prefix, and $y, y' \in \mathcal{Y}$ are the analogous bounds for $P[\ell+1..m]$ (if any of the four does not exist, $\rect_{\ell}$ is empty).

\subparagraph{Reducing to $\rho$ rectangles.}
Because the popped sequence is a core of $P$, only split positions at its run boundaries or at the end of its first nonterminal can yield a nonempty $\pOcc(P, \ell)$, as formalized below.
\begin{lemma}[Generalization of Lemma 7 in \cite{DBLP:journals/dam/NishimotoIIBT20}]\label{lem:core_rectangle}
Let $Q$ be a core of a pattern $P$ of length $m \geq 2$, and let $Q[1]$ be the first nonterminal of $Q$.
Also, let $\rho$ be the number of runs in the run-length encoding of $Q$.
For each $i \in \{1, \ldots, \rho-1\}$, let $\phi(i)$ denote the length of the string derived from the first $i$ runs.
Then, for any integer $\ell \in \{1, 2, \ldots, m-1\}$,
if $\ell \not \in \{ |\val(Q[1])| \} \cup \{ \phi(i) \mid i = 1, 2, \ldots, \rho - 1 \}$,
then $\pOcc(P, \ell) = \emptyset$.
%where, $Q[1]$ is the first nonterminal of $Q$.
\end{lemma}
\begin{proof}
See Appendix~\ref{app:core_rectangle}.   
\end{proof}
Applying Lemma~\ref{lem:core_rectangle} to the popped sequence $Q$, 
this step constructs at most $\rho$ distinct rectangles $\rect_{|\val(Q[1])|}, \rect_{\phi(1)}, \rect_{\phi(2)}, \ldots, \rect_{\phi(\rho-1)}$.
%this step constructs only the $\rho$ rectangles $\rect_{|\val(Q[1])|}, \rect_{\phi(1)}, \rect_{\phi(2)}, \ldots, \rect_{\phi(\rho-1)}$.

\subsection{Step (iii): Range Reporting for Primary Occurrences}
We issue one range-reporting query using Blelloch's data structure for $\mathcal{P}$ (Section~\ref{subsec:data_structures}) for each of these at most $\rho$ distinct rectangles;
%We apply a range reporting query using Blelloch's data structure for $\mathcal{P}$ (Section~\ref{subsec:data_structures}) to each of the $\rho$ rectangles; 
since $\pOcc(P, \ell)$ coincides with the points in $\rect_\ell$, 
taking the union of the reported pairs $(u_i,\ell)$ over the considered split positions yields the complete set of primary occurrences.
%taking the union yields the complete set $\bigcup_{\ell=1}^{m-1} \pOcc(P, \ell)$ of primary occurrences. 

\subsection{Step (iv): Converting Primary Occurrences to \texorpdfstring{$\Occ(T, P)$}{Occ(T, P)}}
For $u_{i} \in \mathcal{U}$, let $\vOcc(u_{i}) \subseteq \{1, 2, \ldots, n\}$ be the starting positions in $T$ of the occurrences of $\mathcal{L}_U(u_i)$ in the derivation tree (i.e., $j \in \vOcc(u_{i})$ if and only if $\mathcal{L}_U(u_i)$ is embedded at position $j$). The following lemma converts primary occurrences into $\Occ(T, P)$ using $\vOcc$, distinguishing two cases by the production rule of $u_i$.
\begin{lemma}[\cite{DBLP:journals/dam/NishimotoIIBT20}]\label{lem:pOcc_and_locate}
Recall that $L_{i}$ is the $x$-coordinate of $u_{i}$ (defined in Section~\ref{subsubsec:range_report}).
For each $\ell \in \{1,2,\ldots,m-1\}$ and each $u_{i} \in \pOcc(P, \ell)$, define $\mathcal{Z}_{\ell}(u_{i})$ according to the production rule of $u_{i}$:
\begin{itemize}
    \item if the rule has the form $X_{i} \rightarrow X_{j} X_{k}$ with two distinct nonterminals,
    \[ \mathcal{Z}_{\ell}(u_{i}) = \{ q + |L_{i}| - \ell \mid q \in \vOcc(u_{i}) \}; \]
    \item if the rule has the form $X_{i} \rightarrow (X_{k})^{d}$,
    \[ \mathcal{Z}_{\ell}(u_{i}) = \Bigl\{ q + j \cdot |L_{i}| - \ell \;\Big|\; q \in \vOcc(u_{i}),\ j \in \{1, 2, \ldots, \lfloor \tfrac{|\val(X_{i})| - (m - \ell)}{|L_{i}|} \rfloor\} \Bigr\}. \]
\end{itemize}
Then,
\[
    \Occ(T, P) = \bigcup_{\ell = 1}^{m-1} \bigcup_{u_{i} \in \pOcc(P, \ell)} \mathcal{Z}_{\ell}(u_{i}).
\]
Moreover, the sets $\mathcal{Z}_{\ell}(u_i)$ are pairwise disjoint over all pairs $(u_i,\ell)$ with $u_i \in \pOcc(P,\ell)$.
%where the ${}+1$ terms convert offsets in $\vOcc$ to positions in $\Occ(T, P)$.
\end{lemma}
By Lemma~\ref{lem:pOcc_and_locate}, taking the union $\bigcup_{\ell = 1}^{m-1} \bigcup_{u_{i} \in \pOcc(P, \ell)} \mathcal{Z}_{\ell}(u_{i})$ yields $\Occ(T, P)$, completing the locate query.

\subsection{Implementation and Complexity Analysis}\label{sec:locate_impl}
This subsection gives the full locate algorithm for $m \geq 1$, after first setting up the auxiliary queries it relies on.

\subsubsection{Auxiliary queries on the DAG}\label{subsubsec:aux_queries}
For the locate algorithm we use two ingredients: Lemma~\ref{lem:basic_queries} for traversing nodes and edges of the DAG, and random access, LCE, and reversed LCE queries on $T$ for computing rectangle coordinates in phase (ii) of Section~\ref{subsubsec:locate2}.
Although inter-path edges are stored without labels (Section~\ref{subsec:data_structures}), each label can be reconstructed in $\mathcal{O}(1)$ time from the source's production rule (Observation~\ref{ob:comp_edges} in Appendix~\ref{app:basic_queries}); each value in $\vOcc(u_i)$ then equals $1$ plus the sum of edge labels along the corresponding root-to-$u_i$ path. The following lemmas summarize the resulting queries.

\begin{lemma}\label{lem:basic_queries}
Using the doubly linked list for $\pathset$, 
we can support the following four operations on a given node $u_i \in \mathcal{U}$ 
lying on a path $\mathbb{P} \in \pathset$:
{\renewcommand{\labelenumi}{(\roman{enumi})}
\begin{enumerate}
\item return the following information in $\mathcal{O}(1)$ time:
(A) the production rule $X_i \rightarrow \expr_i$ corresponding to $u_i$,
(B) $|\val(X_i)|$,
(C) the height of $u_i$,
(D) $\assign(X_i)$, and
(E) the distinct children of $u_i$;
\item return the outgoing edges of $u_i$ with edge labels in $\mathcal{O}(1)$ time per edge;
\item return the directed edges with edge labels in $\mathcal{E}_{\expl}(\mathbb{P})$ in $\mathcal{O}(1)$ time per edge;
\item return $\vOcc(u_i)$ in $\mathcal{O}(H|\vOcc(u_i)|)$ time if $u_i$ is explicit.
\end{enumerate}
}
\end{lemma}
\begin{proof}
See Appendix~\ref{app:basic_queries}.
\end{proof}

\subparagraph{Random access, LCE, and reversed LCE queries.}
A random access query returns $T[i]$; an LCE (resp.\ reversed LCE) query returns the length of the longest common prefix (resp.\ suffix) of $T[i..n]$ and $T[j..n]$ (resp.\ $T[1..i]$ and $T[1..j]$); one argument may instead be a suffix (resp.\ prefix) of $P$ once the popped sequence of $P$ has been computed.
Kempa and Kociumaka's algorithms (Theorems 5.24 and 5.25 in \cite{DBLP:journals/corr/abs-2308-03635}) traverse the derivation tree of the restricted recompression RLSLP to answer these in $\mathcal{O}(H)$ time; Lemma~\ref{lem:basic_queries} enables the same traversal in our representation, and a modification computes a substring of length $\ell$ starting at position $i$ in $\mathcal{O}(H+\ell)$ time.
Phase (ii) of Section~\ref{subsubsec:locate2} uses them to compare, via binary search, a suffix of $P$ 
with the coordinate $R_{s}$ of a candidate node $u_{s}$ (or the reverse of a prefix of $P$ with $L_{s}$) in $\mathcal{O}(H + \log M)$ time per comparison. 
See Appendix~\ref{app:lce} for details of the three queries.

\subsubsection{Full locate algorithm and complexity bound}\label{subsubsec:locate2}
Combining Steps (i)--(iv) with the auxiliary queries of Section~\ref{subsubsec:aux_queries} yields the following.

\begin{theorem}\label{thm:locate}
Given a pattern $P$ of length $m \geq 1$, a locate query returning $\Occ(T, P)$ can be answered in expected 
\[
\mathcal{O}\bigl(m + \log m (H + \log M)\log M + \occ(H + \log M / \log \log M)\bigr)
\]
time using the dynamic RR-index. 
\end{theorem}
\begin{proof}
If $m = 1$, write $P = c$ with $c \in \Sigma$. In $\mathcal{O}(1)$ time, the global hash table tells us whether some $X_i \in \mathcal{V}$ satisfies $X_i \rightarrow c$ (and returns it if so); otherwise $\Occ(T, P) = \emptyset$. The corresponding DAG node $u_i$ is explicit, so Lemma~\ref{lem:basic_queries}(iv) enumerates $\Occ(T, P) = \vOcc(u_i)$ in $\mathcal{O}(\occ \cdot H)$ time, which is absorbed into the stated bound.

For $m \geq 2$, we execute the four phases corresponding to Steps (i)--(iv).

\emph{Phase (i).}
Compute the popped sequence $Q$ of $P$, whose run-length encoding has $\rho$ runs, by the algorithm of~\cite{DBLP:journals/mst/DuysterK26} (Appendix~\ref{app:pop}) in expected $\mathcal{O}(m)$ time.
If it fails, $\Occ(T, P) = \emptyset$ and we return $\emptyset$.

\emph{Phase (ii).}
Compute the $\rho$ rectangles $\rect_{|\val(Q[1])|}, \rect_{\phi(1)}, \ldots, \rect_{\phi(\rho-1)}$ by binary search on $\mathcal{X}$ and $\mathcal{Y}$ for each of the four coordinates. 
Each of the $\mathcal{O}(\log M)$ comparisons takes $\mathcal{O}(H + \log M)$ time (Section~\ref{subsubsec:aux_queries}), 
so this phase runs in $\mathcal{O}(\rho(H + \log M)\log M)$ time with $\mathbb{E}[\rho] = \mathcal{O}(\log m)$.

\emph{Phase (iii).}
Let $\Lambda = \{|\val(Q[1])|\} \cup \{\phi(1),\ldots,\phi(\rho-1)\}$. 
Here, $|\val(Q[1])|$ is a valid split position (i.e., $|\val(Q[1])| \in \{ 1, 2, \ldots, m-1 \}$) 
because $|Q| \geq 2$ for $m \geq 2$ (see the fourth property of the popped sequence in Appendix~\ref{app:popped_sequence}). 
Issue one range-reporting query on $\mathcal{P}$ for each $\ell \in \Lambda$,
and attach $\ell$ to every reported node. This obtains
\[
\{(u_i,\ell) \mid \ell \in \Lambda,\ u_i \in \pOcc(P,\ell)\}
\]
in $\mathcal{O}(\rho \log M + \occ_{\mathsf{prim}}(\log M / \log \log M))$ time,
where $\occ_{\mathsf{prim}} \leq \occ$ is the number of primary occurrences.

%Issue $\rho$ range reporting queries on $\mathcal{P}$ to obtain $\bigcup_{\ell=1}^{m-1} \pOcc(P, \ell)$ in $\mathcal{O}(\rho \log M + \occ_{\mathsf{prim}}(\log M / \log \log M))$ time, where $\occ_{\mathsf{prim}} \leq \occ$ is the number of primary occurrences.

\emph{Phase (iv).}
For each $(u_{i}, \ell)$ with $u_{i} \in \pOcc(P, \ell)$, compute $\vOcc(u_{i})$ in $\mathcal{O}(H|\vOcc(u_{i})|)$ time (Lemma~\ref{lem:basic_queries}(iv)) and return $\mathcal{Z}_{\ell}(u_{i})$ as defined in Lemma~\ref{lem:pOcc_and_locate}.
Since $\sum_{\ell}\sum_{u_{i} \in \pOcc(P, \ell)} |\mathcal{Z}_{\ell}(u_{i})| = \occ$, phase (iv) takes $\mathcal{O}(\occ \cdot H)$ time.

Summing the four phase bounds and absorbing the $\occ_{\mathsf{prim}}(\log M / \log \log M)$ term of phase (iii) into phase (iv) via $\occ_{\mathsf{prim}} \leq \occ$ gives the stated complexity.
\end{proof}

%% file: tex/4_1_improvement.tex
\subsection{Faster Locate via Ancestor-Path Caching}\label{sec:improvement}
\begin{figure}[t]
\makebox[\linewidth][c]{%
		\includegraphics[scale=0.5]{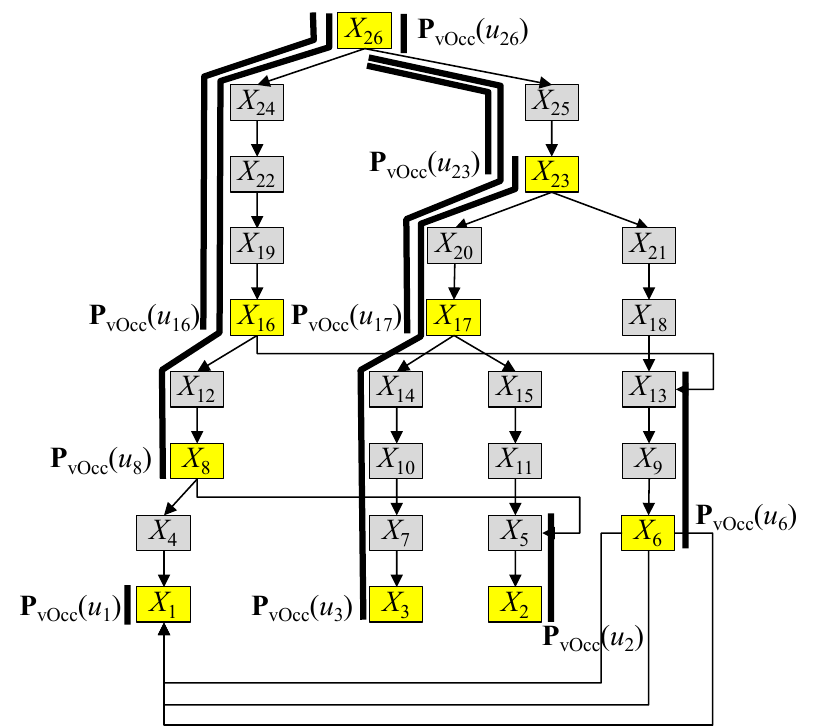}
}
\caption{Each path $\mathbb{P}_{\vOcc}(u_{i})$ on the DAG of Figure~\ref{fig:rrdag} with $\lceil \alpha + \log B \rceil = 6$: the bold line starting at $X_{i}$ is $\mathbb{P}_{\vOcc}(u_{i})$, and $u_{i}$ is the explicit node for $X_{i}$.}
\label{fig:vOccPath}
\end{figure}

In the locate algorithm of Theorem~\ref{thm:locate}, phase (iv) spends $\mathcal{O}(H)$ time per reported occurrence, so the term $\occ \cdot H$ dominates the running time whenever $\occ$ is large.
We reduce this factor to $H/\log B$ by caching, at each explicit node, a short ancestor path in the DAG. 
This yields the main result of this section.

\begin{theorem}\label{thm:locate_improved}
Given a pattern $P$ of length $m \geq 1$, a locate query returning $\Occ(T, P)$ can be answered in expected
\[
\mathcal{O}\bigl(m + \log m \log^{2} n + \occ (\log n / \log \log n) \bigr)
\]
time using the dynamic RR-index augmented with the additional data structures described below. 
This bound holds under the assumption that $H = \mathcal{O}(\log n)$. 
\end{theorem}

Section~\ref{subsubsec:improved_phase4} proves Theorem~\ref{thm:locate_improved} by replacing the algorithm for computing $\vOcc(u_{i})$ (i.e., Lemma~\ref{lem:basic_queries}(iv)) with an improved algorithm that runs in $\mathcal{O}(|\vOcc(u_{i})| \lceil H / \log B \rceil)$ time. 

We also propose a complementary heuristic for phase (ii); 
it does not change the asymptotic complexity but is expected to reduce the running time in practice. 
See Appendix~\ref{app:practical_phase2} for details.

\subparagraph{Additional data structures.}
Fix tunable constants $\alpha, \beta = \mathcal{O}(1)$ controlling the length of a cached ancestor path and a stored coordinate prefix, respectively.
For each path $\mathbb{P} \in \pathset$ with explicit endpoint $u_{i} \in \mathcal{U}_{\expl}$, we augment the corresponding list element with two $\mathcal{O}(B)$-bit fields:
(i) the pair $(u_j, W)$ encoding the longest path $\mathbb{P}_{\vOcc}(u_{i})$ ending at $u_i$ with $|\mathbb{P}_{\vOcc}(u_{i})| \leq \lceil \alpha + \log B \rceil$ such that every node except its first node $u_{j}$ has exactly one parent in $\dagraph$, where $W$ is the sum of edge labels along the path (if $u_i$ does not have exactly one parent, $\mathbb{P}_{\vOcc}(u_i) = (u_i)$ with $u_j = u_i$ and $W = 0$); this field is used in Section~\ref{subsubsec:improved_phase4}.
%(ii) The first $\lceil \beta B / \log \sigma \rceil$ characters of $L_i$ and of $R_i$; used by the phase (ii) heuristic of Appendix~\ref{app:practical_phase2}.
(ii) the prefixes of $L_i$ and $R_i$ of length at most $\lceil \beta B / \log \sigma \rceil$; this field is used in the phase-(ii) heuristic.
Figure~\ref{fig:vOccPath} illustrates $\mathbb{P}_{\vOcc}(u_i)$.
The $M$ list elements together add $\mathcal{O}(MB)$ bits, so the doubly linked list still fits in $\mathcal{O}(M(H+B))$ bits.

\subsubsection{Speeding up computation of \texorpdfstring{$\vOcc(u_i)$}{vOcc}}\label{subsubsec:improved_phase4}
%\subsubsection{Speeding up phase (iv)}\label{subsubsec:improved_phase4}
The cached paths $\mathbb{P}_{\vOcc}(u_{i})$ let us compute $\vOcc(u_i)$ via the following recurrence.
\begin{lemma}\label{lem:additional_vOcc}
For an explicit node $u_{i} \in \mathcal{U}_{\expl}$, let $(u_{j}, W)$ be the encoding of $\mathbb{P}_{\vOcc}(u_{i})$ and let $\mathbb{P}_{j} \in \pathset$ be the path containing $u_{j}$. Then
\[
\vOcc(u_{i}) =
\begin{cases}
\bigcup_{e \in \mathcal{E}_{\expl}(\mathbb{P}_{j})} \{ W + \mathcal{L}_{E}(e) + q \mid q \in \vOcc(\mathrm{src}(e)) \} & \text{if } \mathcal{E}_{\expl}(\mathbb{P}_{j}) \neq \emptyset, \\
\{ W+1 \} & \text{otherwise.}
\end{cases}
\]
\end{lemma}
\begin{proof}
Lemma~\ref{lem:additional_vOcc} follows from Observation~\ref{ob:comp_edges} in Appendix~\ref{app:basic_queries} and the property of edge labels stated in Section~\ref{subsubsec:aux_queries}.
\end{proof}

Guided by Lemma~\ref{lem:additional_vOcc}, the algorithm reads $(u_{j}, W)$ from the list element, follows the pointer at $u_j$ to the element representing $\mathbb{P}_j$, recovers $\mathcal{E}_{\expl}(\mathbb{P}_{j})$ with its edge labels via Lemma~\ref{lem:basic_queries}(iii), and either returns $\{W+1\}$ (if empty) or recursively computes $\vOcc(\mathrm{src}(e))$ and returns the union in the recurrence.
The total number $Z$ of recursive calls satisfies $Z = \mathcal{O}(|\vOcc(u_i)| \lceil H / \log B \rceil)$, since they decompose every root-to-$u_i$ path in the DAG into cached paths, inter-path edges, and zero-label paths. See Appendix~\ref{app:z_bound} for details. 
Substituting into Theorem~\ref{thm:locate} with $\log M = \mathcal{O}(\log n)$ and $B = \mathcal{O}(\log n)$ yields Theorem~\ref{thm:locate_improved} under $H = \mathcal{O}(\log n)$.

%% file: tex/7_update.tex
\section{Update Operations}\label{sec:update}
%\subsection{Insertion Operation}
\subsection{Insertion Operation}
%\label{sec:string_insertion}

Given a string $P' \in \Sigma^{+}$ of length $m' = \mathcal{O}(n)$ and a position $s \in \{1, 2, \ldots, n + 1\}$, the insertion operation produces $T'$ by inserting $P'$ into $T$ at position $s$ and updates the dynamic RR-index so that its RLSLP derives $T'$.
For simplicity, we focus on $2 \leq s \leq n$, so that $T'$ is the concatenation of $T[1..s-1]$, $P'$, and $T[s..n]$; the boundary cases $s = 1$ and $s = n+1$ are analogous.
Following the standard strategy of dynamic grammar-compressed data structures~\cite{DBLP:journals/iandc/NishimotoTT20,DBLP:journals/dam/NishimotoIIBT20,DBLP:journals/corr/abs-2404-07510,DBLP:conf/mfcs/TanimuraNBIT17}, the update proceeds in three phases (see Appendix~\ref{app:insertion_operation} for details).

\subparagraph{Phase 1: constructing the derivation tree of $\mathcal{G}^{R'}$.}
The first phase constructs the derivation tree of an RLSLP $\mathcal{G}^{R'} = (\mathcal{V}', \Sigma', \mathcal{D}', E')$ deriving $T'$ by restricted recompression, computing $(H'+1)$ sequences $S'^{0}, S'^{1}, \ldots, S'^{H'}$ bottom-up; production rules are shared with $\mathcal{G}^{R}$ whenever the right-hand sides agree.
For each nonterminal in $\mathcal{V}' \setminus \mathcal{V}$, we insert the corresponding node into the DAG and update the dynamic data structures.
A naive construction takes $\Omega(n + m')$ time since the tree has $n + m'$ leaves. We reduce this cost using the popped sequences $Q_{L}$ and $Q_{R}$ of $T[1..s-1]$ and $T[s..n]$: since each popped sequence is a core of its source, $Q_{L}$ and $Q_{R}$ are embedded in the derivation tree of $\mathcal{G}^{R'}$ at positions $1$ and $s+m'$ as subtrees that need not be reconstructed.
Concretely, (i) compute $Q_{L}$ and $Q_{R}$ by the algorithm of~\cite{DBLP:journals/mst/DuysterK26} in $\mathcal{O}(H)$ time; (ii) build $S'^{0}, \ldots, S'^{H'}$ in $\mathcal{O}(|\mathcal{T}_{\diff}|)$ time, where $\mathcal{T}_{\diff}$ collects (A) the ancestors of the embedded subtree roots, (B) the $m'$ leaves covering $T'[s..s+m'-1]$, and (C) their ancestors. Figure~\ref{fig:update_derivation_tree} illustrates the embeddings.
The phase runs in expected amortized $\mathcal{O}(|\mathcal{T}_{\diff}| + |\mathcal{I}_{\expl}| \max\{H, H', \log M\} \log M + |\mathcal{I}_{\impl}|)$ time, where $\mathcal{I}_{\expl}$ (resp.\ $\mathcal{I}_{\impl}$) is the set of explicit (resp.\ implicit) nodes inserted into the DAG.

\subparagraph{Phase 2: removing obsolete nonterminals.}
For each nonterminal in $\mathcal{V} \setminus \mathcal{V}'$, we remove the corresponding node from the DAG and update the dynamic data structures.
This phase takes expected amortized $\mathcal{O}(|\mathcal{R}_{\expl}| \log M + |\mathcal{R}_{\impl}|)$ time.
Here, $\mathcal{R}_{\expl}$ (resp.\ $\mathcal{R}_{\impl}$) is the set of explicit (resp.\ implicit) nodes whose nonterminals do not appear in the derivation tree of $\mathcal{G}^{R'}$.

\subparagraph{Phase 3: updating ancestor-path caches and coordinate prefixes.}
For each explicit node in $\mathcal{I}_{\expl} \cup \mathcal{U}_{\change}$, we update the additional information stored in the corresponding element of the doubly linked list for $\pathset$.
Here, $\mathcal{U}_{\change} \subseteq \mathcal{U}_{\expl}$ is the set of explicit nodes that have an ancestor in $\mathcal{I}_{\expl} \cup \mathcal{I}_{\impl} \cup \mathcal{R}_{\expl} \cup \mathcal{R}_{\impl}$ within distance $\lceil \alpha + \log B \rceil$ in the DAG. 
This phase takes $\mathcal{O}(W_{\max} + (|\mathcal{U}_{\change}| + |\mathcal{I}_{\expl}|) (H' + B))$ time,
where $W_{\max} = \max \{ |\mathcal{I}_{\expl}|, |\mathcal{I}_{\impl}|, |\mathcal{R}_{\expl}|, |\mathcal{R}_{\impl}|, |\mathcal{U}_{\change}| \}$.

\begin{figure}[t]
\makebox[\linewidth][c]{%
		\includegraphics[scale=0.55]{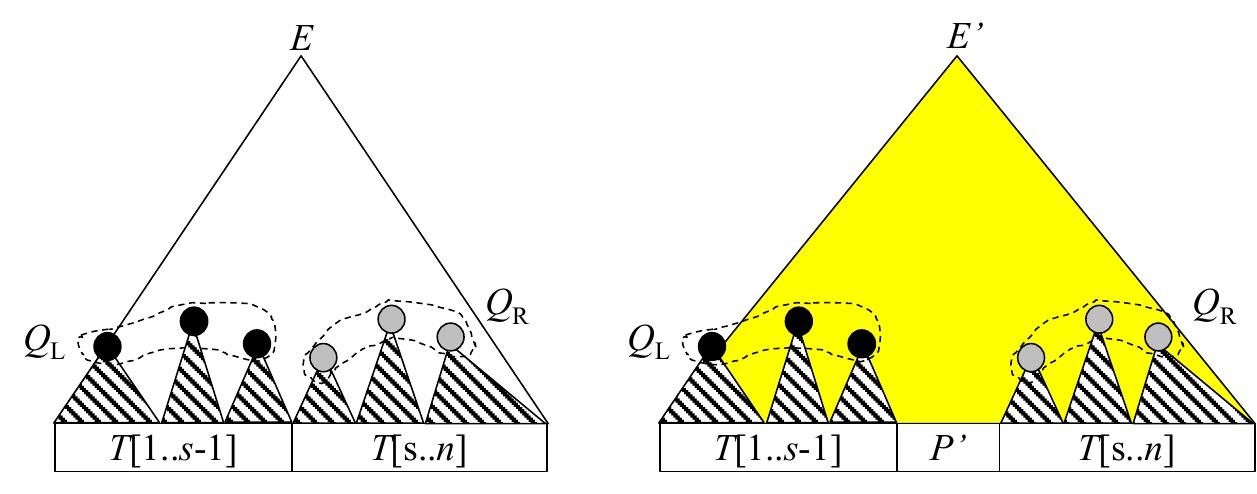}
}
\caption{Embeddings of $Q_{L}$ and $Q_{R}$ in the derivation trees of $\mathcal{G}^{R}$ and $\mathcal{G}^{R'}$ (start symbols $E$ and $E'$). The black/gray circles mark the nodes for $Q_{L}/Q_{R}$, hatched regions are the subtrees rooted there, and the yellow region is the part of the derivation tree of $\mathcal{G}^{R'}$ to be constructed.}
\label{fig:update_derivation_tree}
\end{figure}

\subparagraph{Correctness and running time.}
To ensure that $H' = \mathcal{O}(\log n)$, 
we declare the update to fail if 
$H' > \lceil 2(w+1)\log_{8/7}(4(n+m'))+2 \rceil$. 
By Lemma~\ref{lem:tree_height}, 
the update succeeds with high probability for a sufficiently large constant $w$.
%By Lemma~\ref{lem:tree_height}, the update fails---i.e., $H' > 2(w+1)\log_{8/7}(4(n+m'))+2$---with probability at most $1/n^{w}$, so it succeeds w.h.p.\ for sufficiently large $w$.
Since $H, H', B, \log M = \mathcal{O}(\log (n+m'))$ and $m' = \mathcal{O}(n)$, summing the three phases yields
\begin{equation*}
    \mathcal{O}(|\mathcal{T}_{\diff}| + (|\mathcal{I}_{\expl}| + |\mathcal{R}_{\expl}|) \log^{2} n + W_{\max} + |\mathcal{U}_{\change}| \log n).
\end{equation*}

The following lemma bounds $|\mathcal{T}_{\diff}|$, $|\mathcal{I}_{\expl}|$, $|\mathcal{I}_{\impl}|$, $|\mathcal{R}_{\expl}|$, $|\mathcal{R}_{\impl}|$, and $|\mathcal{U}_{\change}|$. 
\begin{lemma}\label{lem:ancestor_upper_bound}
(i) $|\mathcal{T}_{\diff}| = \mathcal{O}((m' + H) H')$, 
(ii) $|\mathcal{I}_{\impl}| = \mathcal{O}((m' + H) H')$, 
(iii) $|\mathcal{R}_{\impl}| = \mathcal{O}(H^{2})$, 
(iv) $|\mathcal{I}_{\expl}| = \mathcal{O}(m' + H)$, 
(v) $|\mathcal{R}_{\expl}| = \mathcal{O}(H)$,
and (vi) $|\mathcal{U}_{\change}| = \mathcal{O}(HB)$.
\end{lemma}
\begin{proof}
See Appendix~\ref{app:ancestor_upper_bound}.
\end{proof}
Therefore, the total expected amortized running time is $\mathcal{O}(m' \log^{2} n + \log^{3} n)$.

\subsection{Deletion Operation}
Given a position $s \in \{1, 2, \ldots, n\}$ and a length $m'$ with $s + m' - 1 \leq n$, the deletion operation produces $T'$ by removing the substring $T[s..s+m'-1]$ from $T$ and updates the dynamic RR-index so that its RLSLP derives $T'$.
It runs in expected amortized $\mathcal{O}(m' \log^{2} n + \log^{3} n)$ time by a similar approach. 
See Appendix~\ref{app:deletion_operation} for details.

%% file: tex/8_experiment.tex
\section{Experiments}\label{sec:experiment}
\newcommand{\Dataset}[1]{\textsf{#1}}

\subparagraph{Setup.}
We evaluated the dynamic RR-index 
%on locate queries and update operations over eleven highly repetitive strings:
for locate queries and update operations on eleven highly repetitive strings:
(i) nine strings from the Pizza\&Chili repetitive corpus~\cite{dataset:pc-repetitive-corpus};
(ii) a 37~GB string (\Dataset{enwiki}) of English Wikipedia articles with complete edit history~\cite{dataset:enwiki-all-pages}; and
(iii) a 59~GB string (\Dataset{chr19}) obtained by concatenating chromosome~19 from 1{,}000 human genomes in the 1000 Genomes Project~\cite{1000Genomes}.
Per-dataset statistics ($\sigma$, $n$, $\delta$, the explicit-node count $M$, 
the height $H$ of the derivation tree, 
and the average occurrence counts $\occ_{10}, \occ_{100}, \occ_{1000}$ for $1{,}000$ patterns of lengths $10$, $100$, and $1000$) appear in Appendix~\ref{app:dataset-stats_full}; for \Dataset{chr19}, $\sigma=4$, $n \approx 5.9 \times 10^{10}$, $\delta \approx 2.7 \times 10^{6}$, $M \approx 2.7 \times 10^{7}$, $H = 372$, $\occ_{10} = 881{,}774$, $\occ_{100} = 940$, and $\occ_{1000} = 504$. 

We compared the dynamic RR-index with both dynamic and static compressed indexes: the dynamic SE-index~\cite{DBLP:journals/dam/NishimotoIIBT20}, the dynamic r-index~\cite{dynamic_r_index_dcc}, and the static r-index~\cite{DBLP:journals/jacm/GagieNP20}.
We implemented the dynamic RR-index and the dynamic SE-index from scratch in C++ (\url{https://github.com/TNishimoto/dynamic-rlslp-index}), and used the public implementations of the r-index~(\url{https://github.com/nicolaprezza/r-index}) and the dynamic r-index~(\url{https://github.com/TNishimoto/dynamic_r_index}). 
All experiments ran on a single core of a 48-core Intel Xeon Gold 6126 CPU (2.6~GHz) with 2~TB of RAM, under 64-bit CentOS 7.9.

\subparagraph{Implementation details of the dynamic RR-index.}
The dynamic RR-index has three parameters: $w$ (Section~\ref{sec:recompression}) and $\alpha, \beta$ (Section~\ref{sec:improvement}).
For performance reasons, we fix the integer quantity that each parameter controls, as follows.
We store node heights in $16$-bit integers, i.e., we require the height bound to satisfy $\lceil 2(w+1)\log_{8/7}(4n)+2 \rceil \le 65535$ $(= 2^{16}-1)$, which fixes $w$.
We set the cached ancestor-path length $\lceil \alpha + \log B \rceil$ to $6$ when $n \le 10^{9}$, and to $8$ otherwise.
We set the stored coordinate-prefix length $\lceil \beta B / \log \sigma \rceil$ to $\lfloor 64 / \log \sigma \rfloor$, so that the first $\lfloor 64 / \log \sigma \rfloor$ characters fit in a single $64$-bit integer.

Since no public implementation of Blelloch's data structure (Section~\ref{subsec:data_structures}) was available, we used a dynamic wavelet-matrix-based data structure~\cite{DBLP:journals/is/ClaudeNP15} for two-dimensional range reporting.
It uses $\mathcal{O}(MB)$ bits of space, supports range reporting queries in $\mathcal{O}((1 + \rocc) \log^{2} M)$ time, where $\rocc$ is the number of reported points, and supports point insertions and deletions in amortized $\mathcal{O}(\log^{2} M)$ time.
Therefore, 
this substitution increases the time complexities of locate queries and update operations by logarithmic factors. 

\begin{figure}[t]
\centering
\begin{minipage}[t]{0.32\linewidth}\centering
    {\small $(m, m') = (10, 10)$}\\
    \includegraphics[width=\linewidth]{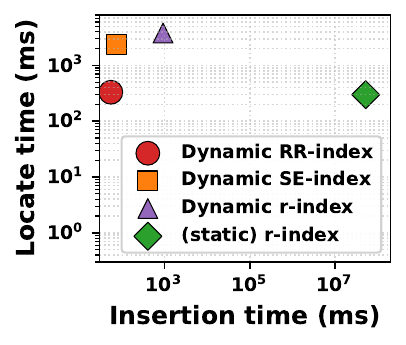}
\end{minipage}\hfill
\begin{minipage}[t]{0.32\linewidth}\centering
    {\small $(m, m') = (100, 100)$}\\
    \includegraphics[width=\linewidth]{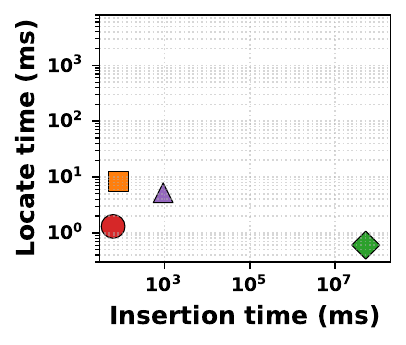}
\end{minipage}\hfill
\begin{minipage}[t]{0.32\linewidth}\centering
    {\small $(m, m') = (1000, 1000)$}\\
    \includegraphics[width=\linewidth]{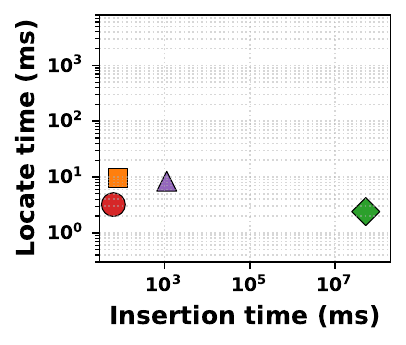}
\end{minipage}
\caption{Insertion--locate trade-off on \Dataset{chr19} for $(m, m') \in \{(10, 10), (100, 100), (1000, 1000)\}$. Each panel plots, for each method, the mean insertion time ($x$-axis) against the mean locate time ($y$-axis), both on a log scale; lower-left is better.}
\label{fig:tradeoff_chr19_ins}
\end{figure}
\begin{figure}[t]
\centering
\begin{minipage}[t]{0.32\linewidth}\centering
    {\small $(m, m') = (10, 10)$}\\
    \includegraphics[width=\linewidth]{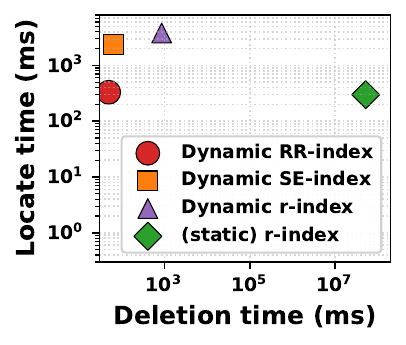}
\end{minipage}\hfill
\begin{minipage}[t]{0.32\linewidth}\centering
    {\small $(m, m') = (100, 100)$}\\
    \includegraphics[width=\linewidth]{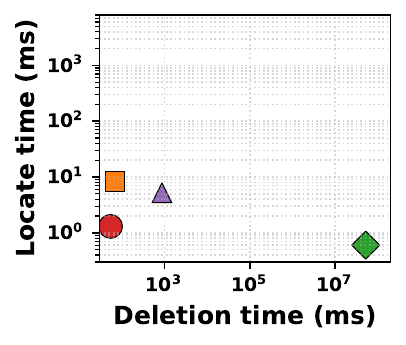}
\end{minipage}\hfill
\begin{minipage}[t]{0.32\linewidth}\centering
    {\small $(m, m') = (1000, 1000)$}\\
    \includegraphics[width=\linewidth]{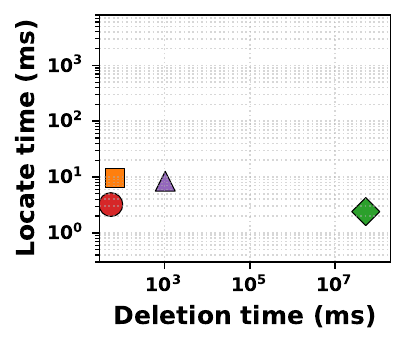}
\end{minipage}
\caption{Deletion--locate trade-off on \Dataset{chr19}. Axes and conventions are as in Figure~\ref{fig:tradeoff_chr19_ins}.}
\label{fig:tradeoff_chr19_del}
\end{figure}

\begin{figure}[t]
\centering
\begin{minipage}[t]{0.32\linewidth}\centering
    {\small $(m, m') = (10, 10)$}\\
    \includegraphics[width=\linewidth]{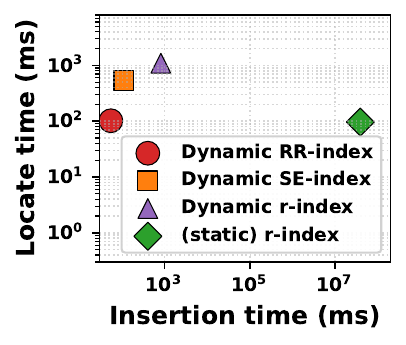}
\end{minipage}\hfill
\begin{minipage}[t]{0.32\linewidth}\centering
    {\small $(m, m') = (100, 100)$}\\
    \includegraphics[width=\linewidth]{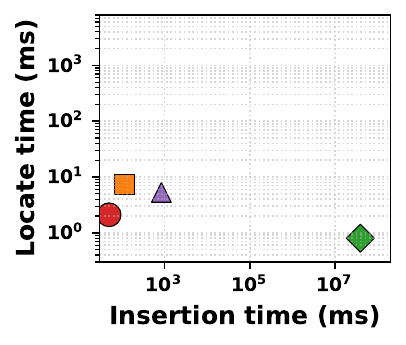}
\end{minipage}\hfill
\begin{minipage}[t]{0.32\linewidth}\centering
    {\small $(m, m') = (1000, 1000)$}\\
    \includegraphics[width=\linewidth]{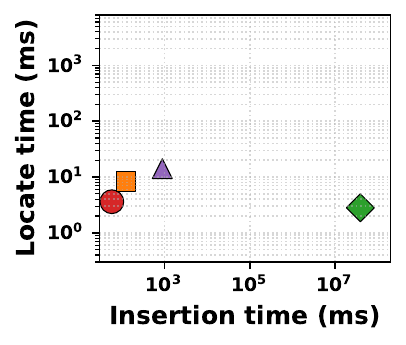}
\end{minipage}
\caption{Insertion--locate trade-off on \Dataset{enwiki}. Axes and conventions are as in Figure~\ref{fig:tradeoff_chr19_ins}.}
\label{fig:tradeoff_enwiki_ins}
%\end{figure}
%\begin{figure}[t]
\centering
\begin{minipage}[t]{0.32\linewidth}\centering
    {\small $(m, m') = (10, 10)$}\\
    \includegraphics[width=\linewidth]{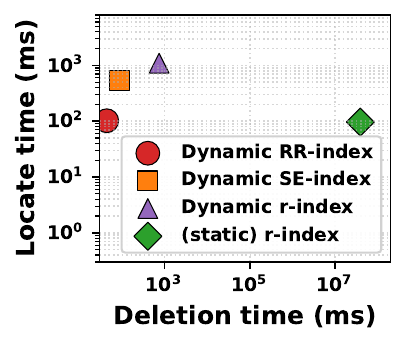}
\end{minipage}\hfill
\begin{minipage}[t]{0.32\linewidth}\centering
    {\small $(m, m') = (100, 100)$}\\
    \includegraphics[width=\linewidth]{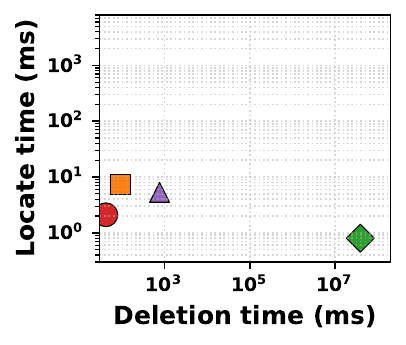}
\end{minipage}\hfill
\begin{minipage}[t]{0.32\linewidth}\centering
    {\small $(m, m') = (1000, 1000)$}\\
    \includegraphics[width=\linewidth]{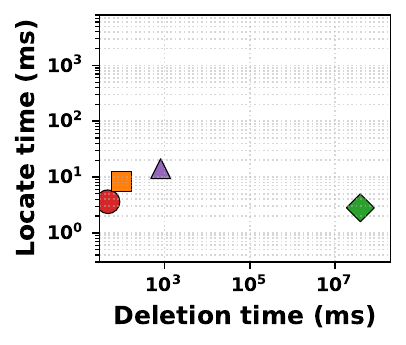}
\end{minipage}
\caption{Deletion--locate trade-off on \Dataset{enwiki}. Axes and conventions are as in Figure~\ref{fig:tradeoff_chr19_ins}.}
\label{fig:tradeoff_enwiki_del}
%\end{figure}
%\begin{figure}[t]
\centering
\begin{minipage}[c]{0.22\linewidth}\centering
    \includegraphics[width=\linewidth]{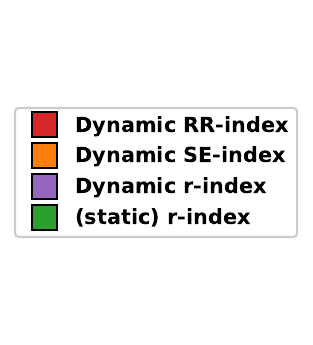}
\end{minipage}\hfill
\begin{minipage}[c]{0.37\linewidth}\centering
    \includegraphics[width=\linewidth]{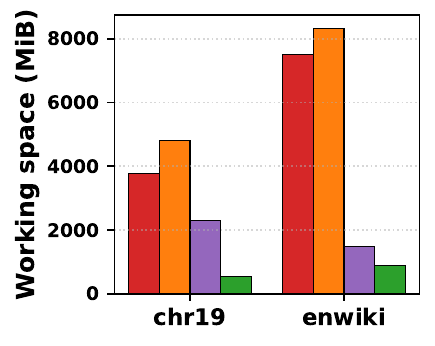}
\end{minipage}\hfill
\begin{minipage}[c]{0.37\linewidth}\centering
    \includegraphics[width=\linewidth]{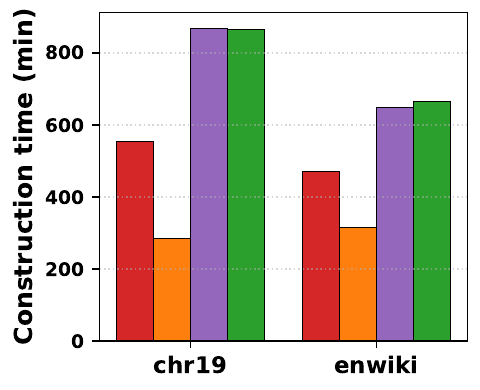}
\end{minipage}
\caption{Working space during locate queries (center) and construction time (right) on \Dataset{chr19} and \Dataset{enwiki}.}
\label{fig:space_construction}
\end{figure}

\subparagraph{Results for \Dataset{chr19}.}
For \Dataset{chr19} we measured update time, locate time, working space, and construction time. 
We inserted, then deleted, $1{,}000$ random substrings of length $m' \in \{1, 10, 100, 1000\}$ at random positions in $T$, and ran $1{,}000$ locate queries with random patterns of length $m \in \{10, 100, 1000\}$. 
In the trade-off figures below, the (static) r-index does not support updates, so its $x$-coordinate is the full reconstruction time ($\approx 5.2 \times 10^{7}$~ms $\approx 14.5$~h on \Dataset{chr19}; $\approx 3.9 \times 10^{7}$~ms $\approx 10.8$~h on \Dataset{enwiki}).

The dynamic RR-index appears in the lower-left region of every panel of Figures~\ref{fig:tradeoff_chr19_ins} and~\ref{fig:tradeoff_chr19_del}: 
it dominates the dynamic SE-index and dynamic r-index on both axes for every $(m, m')$, 
running roughly $1.3\times$ faster than the dynamic SE-index and $15$--$19\times$ faster than the dynamic r-index for update operations, 
and at least $2.6\times$ faster than both for locate queries (reaching $\approx 7\times$ vs.\ the dynamic SE-index and $\approx 11\times$ vs.\ the dynamic r-index at $m = 10$). 
%On locate, 
For locate queries, 
it is competitive with the (static) r-index---within $1.1$--$1.3\times$ at $m \in \{10, 1000\}$ and within $2.2\times$ at $m = 100$---while additionally supporting updates that the r-index cannot. 
On the $59$~GB \Dataset{chr19} ($\delta \approx 2.7 \times 10^{6}$), the dynamic RR-index used $\approx 3.7$~GB of working space (Figure~\ref{fig:space_construction}, center)---about $1/16$ of the raw text size---and finished construction in $\approx 9$~h (Figure~\ref{fig:space_construction}, right), within commodity-server budgets. 
Full numbers appear in Appendix~\ref{app:result_other_datasets}. 

\subparagraph{Results for \Dataset{enwiki}.}
The same trade-off pattern holds on \Dataset{enwiki} (Figures~\ref{fig:tradeoff_enwiki_ins} and~\ref{fig:tradeoff_enwiki_del}): the dynamic RR-index dominates the other two dynamic indexes, running roughly $2.1\times$ faster than the dynamic SE-index and $16\times$ faster than the dynamic r-index for update operations, and is competitive with the (static) r-index for locate queries---matching it for short patterns and within $\approx 2.6\times$ for medium patterns. On the $37$~GB \Dataset{enwiki} ($\delta \approx 7.3 \times 10^{6}$), the dynamic RR-index used $\approx 7.3$~GB of working space (Figure~\ref{fig:space_construction}, center) and finished construction in $\approx 8$~h (Figure~\ref{fig:space_construction}, right), again within commodity-server budgets; 
the $\approx 2\times$ memory usage for a $\delta$ value approximately $2.7\times$ larger than that of \Dataset{chr19} is consistent with linear-in-$\delta$ scaling. 
Full numbers appear in Appendix~\ref{app:result_other_datasets}.

\subparagraph{Results for the Pizza\&Chili corpus.}
Across most settings in the Pizza\&Chili corpus, the dynamic RR-index was either the fastest dynamic index or competitive with the fastest one for update operations, 
reaching a speedup of $\approx 77\times$ over the dynamic r-index on \Dataset{coreutils} at $m' = 1000$
when insertions and deletions are considered together. 
For locate queries, the dynamic RR-index was usually the fastest among the dynamic
indexes, although there were a few exceptions.
Compared with the static r-index, it was within $2.5\times$ for long patterns, but up to
about $8.6\times$ slower for short patterns on \Dataset{world leaders}.
Its working space scaled with $\delta$, from $19$~MiB ($\delta \approx 1.6 \times 10^4$, \Dataset{einstein.de.txt}) to $886$~MiB ($\delta \approx 1.4 \times 10^6$, \Dataset{para}); construction finished within $7$ minutes on every Pizza\&Chili string.

\subparagraph{Summary.}
The experiments show that the dynamic RR-index (i) scales with $\delta$ in memory, (ii) outperforms the other 
dynamic compressed indexes for both updates and locate queries in most settings, 
%dynamic compressed indexes on both updates and locate in most settings, 
and (iii) is competitive with the (static) r-index for short and long patterns on large datasets.

%% file: tex/9a_appendix.tex
\section{Summary Comparison of Dynamic and Static Self-Indexes}\label{app:summarize_table}
Table~\ref{tab:dat} summarizes state-of-the-art dynamic and static self-indexes supporting locate queries.
\begin{table}[p] 
\caption{Summary of state-of-the-art dynamic and static self-indexes supporting locate queries.
Here, $n$ is the length of the input string $T$, $m$ is the length of the pattern $P$, 
$\delta$ is the substring complexity of $T$, 
$\occ$ is the number of occurrences of $P$ in $T$, and $m'$ is the length of an inserted string or a deleted substring.
Moreover, $\occ_c \geq \occ$ is the number of candidate occurrences returned by the ESP-index~\cite{DBLP:journals/jda/MaruyamaNKS13}, $\sigma = n^{\mathcal{O}(1)}$ is the alphabet size, 
$\epsilon > 0$ is an arbitrary constant, and $L_{\max} \leq n$ is the maximum value in the LCP array of $T$.
Let $r$ be the number of runs in the BWT~\cite{burrows1994block} of $T$, 
where $r = \mathcal{O}(\delta \log \delta \max \{ 1, \log (n / (\delta \log \delta)) \})$~\cite{DBLP:conf/focs/KempaK20}.
Let $z$ be the number of factors in the LZ77 factorization~\cite{LZ76} of $T$, 
where $z = \mathcal{O}(\delta \log \frac{n \log \sigma}{\delta \log n})$~\cite{9961143}.
Let $g$ be the size of a compressed grammar deriving $T$, where $g = \mathcal{O}(z \log n \log^{*} n)$~\cite{DBLP:journals/dam/NishimotoIIBT20}. 
Let $\gamma$ be the size of the smallest string attractor~\cite{DBLP:conf/stoc/KempaP18} for $T$, 
where $\gamma = \mathcal{O}(\delta \log \frac{n \log \sigma}{\delta \log n})$~\cite{9961143}.
Also, $q \geq 1$ is a user-defined parameter, and $g'$ is the total size of the $q$-truncated suffix tree of $T$ and a compressed grammar deriving $T$, where $g' = \mathcal{O}(z(q^2 + \log n \log^{*} n))$~\cite{DBLP:journals/iandc/NishimotoTT20}. 
We assume a machine word size $B = \Theta(\log n)$ and $m' = \mathcal{O}(n)$. 
We exclude dynamic self-indexes that use $\Omega(n \log n)$ bits, such as \cite{DBLP:conf/soda/AlstrupBR00,DBLP:conf/cpm/EhrenfeuchtMW09}, with the exception of \cite{DBLP:journals/corr/GawrychowskiKKL15}. 
Furthermore, we exclude data structures that do not explicitly support locate queries, 
even when locate queries can be answered by combining other supported queries, 
because the resulting locate-query time is slower than that of standard self-indexes (e.g., \cite{DBLP:conf/stoc/KempaK22,DBLP:conf/focs/KempaK23,DBLP:journals/corr/abs-2404-07510}).
%Furthermore, we exclude data structures that do not explicitly support locate queries, 
%even if locate queries can be answered by combining other supported queries, 
%This is because their locate-query time is slower than that of standard self-indexes (e.g., \cite{DBLP:conf/stoc/KempaK22,DBLP:conf/focs/KempaK23,DBLP:journals/corr/abs-2404-07510}).
%The abbreviations $\DFMI$, $\SEI$, $\DRI$, and $\RRI$ represent the dynamic FM-index, dynamic SE-index, dynamic r-index, and dynamic RR-index, respectively.
}
\label{tab:dat}

%\footnotesize
%\centering
\scriptsize
%\begin{center}
%\resizebox{\textwidth}{!}{
%\scalebox{0.5}[0.9]{

\begin{tabular}{l|c|c|c}
Method  & Format & Type & Space (bits) \\
\hline
\hline
Dynamic FM-index~\cite{DBLP:journals/jda/SalsonLLM10,DBLP:journals/jda/LeonardMS12} & BWT & Dynamic & $\mathcal{O}(n \log \sigma)$ \\
\hline
Dynamic SE-index~\cite{DBLP:journals/dam/NishimotoIIBT20} & Grammar & Dynamic & $\mathcal{O}(g \log n)$ \\
\hline
TST-index-d~\cite{DBLP:journals/iandc/NishimotoTT20} & Grammar & Dynamic & $\mathcal{O}(g^{\prime} \log n)$ \\
\hline
Gawrychowski et al.~\cite{DBLP:conf/soda/GawrychowskiKKL18,DBLP:journals/corr/GawrychowskiKKL15} & Grammar & Dynamic & $\Omega(n \log n)$ \\
\hline
Dynamic r-index~\cite{dynamic_r_index_dcc} & RLBWT & Dynamic & $\mathcal{O}(r \log n)$ \\
\hline
\hline
r-index~\cite{DBLP:journals/jacm/GagieNP20} & RLBWT & Static & $\mathcal{O}(r \log n)$ \\
\hline
OptBWTR~\cite{DBLP:conf/icalp/NishimotoT21} & RLBWT & Static & $\mathcal{O}(r \log n)$ \\
\hline
Christiansen et al.~\cite{DBLP:journals/talg/ChristiansenEKN21} & Grammar & Static & $\mathcal{O}(\gamma \log (n/\gamma) \log n)$ \\
\hline
Kociumaka et al.~\cite{ViceVersa} & Grammar & Static & $\mathcal{O}(\delta \log \frac{n \log \sigma}{\delta \log n} \log n)$ \\
\hline
\hline
Dynamic RR-index (this paper) & Grammar & Dynamic & expected $\mathcal{O}(\delta \log \frac{n \log \sigma}{\delta \log n} \log n)$ \\
\end{tabular}

%%%%%%%%%%%%%%%%%%%%%%%%%%%%%%%%%%%%%%%%%%%%%%%%%%%%%%%%%%%%%%
\vspace{5mm}

\begin{tabular}{l|c|c}
Method & Locate time & Update time \\
\hline
\hline
Dynamic FM-index~\cite{DBLP:journals/jda/SalsonLLM10,DBLP:journals/jda/LeonardMS12} & $\mathcal{O}((m+\occ) \log^{2+\epsilon} n)$ & $\mathcal{O}((m^{\prime} + L_{\max}) \log n)$ \\
\hline
\multirow{3}{*}{Dynamic SE-index~\cite{DBLP:journals/dam/NishimotoIIBT20}} & $\mathcal{O}(m (\log \log n)^{2}$ & amortized \\
                                                                           & $+\log m (\log n \log^{*} n)^{2}$ & $\mathcal{O}(m^{\prime}(\log n \log^{*} n)^{2}$ \\
                                                                           & $+ \occ \log n)$ & $ + \log n(\log n \log^{*} n)^{2})$ \\
\hline
\multirow{3}{*}{TST-index-d~\cite{DBLP:journals/iandc/NishimotoTT20}} & $\mathcal{O}(m (\log\log \sigma)^{2} + \occ \log n)$ ($m \leq q$) & $\mathcal{O}((\log \log n)^{2} (m^{\prime}q + q^{2} +$ \\ 
\cline{2-2}
                                                                      &  $\mathcal{O}(m (\log\log n)^{2}$ & $ \log n \log^{*} n))$ \\ 

                                                                      & $+ \occ_{c} \log m \log n \log^{*} n)$ ($m > q$) &  \\
\hline
Gawrychowski et al.~\cite{DBLP:journals/corr/GawrychowskiKKL15} & $\mathcal{O}(m + \log^{2} n + \occ \log n)$ & $\mathcal{O}(m^{\prime} \log n + \log^{2} n)$ \\
\hline
Dynamic r-index~\cite{dynamic_r_index_dcc} & $\mathcal{O}((m+\occ) \log n)$ & $\mathcal{O}((m^{\prime} + L_{\max}) \log n)$ \\
\hline
\hline
r-index~\cite{DBLP:journals/jacm/GagieNP20} & $\mathcal{O}(m \log \log \sigma + \occ \log \log n)$ & Unsupported \\
\hline
OptBWTR~\cite{DBLP:conf/icalp/NishimotoT21} & $\mathcal{O}(m \log \log \sigma + \occ)$ & Unsupported \\
\hline
Christiansen et al.~\cite{DBLP:journals/talg/ChristiansenEKN21} & $\mathcal{O}(m + (1 + \occ) \log^{\epsilon} n)$ & Unsupported \\
\hline
Kociumaka et al.~\cite{ViceVersa} & $\mathcal{O}(m + (1 + \occ) \log^{\epsilon} n)$ & Unsupported \\
\hline
\hline
Dynamic RR-index & expected                            & expected amortized \\
(this paper)     & $\mathcal{O}(m + \log m \log^{2} n$ & $\mathcal{O}(m' \log^{2} n + \log^{3} n)$ \\
                 & $+  \occ (\log n / \log \log n))$   &  \\
\end{tabular}
\end{table}

%%%%%%%%%%%%%%%%%%%%%%%%%%%%%%%%%%%%%%%%%%%%%%%%%%%%%%%%%%%%%%%%%%%%%%
\section{Details of Restricted Recompression}\label{app:restricted_recompression}
\subsection{Proof of Lemma~\ref{lem:tree_height}}\label{app:tree_height}
\begin{proof}
Fix integers $w \geq 1$ and $h \geq 0$, and set $|S^{h}| := 1$ for $h > H$; under this extension, $|S^{h}| = 1 \iff H \leq h$.
By Lemma~V.10 of~\cite{9961143}, $\mathbb{E}[|S^{h}|] < 1 + \tfrac{4n}{\mu(h+1)}$, so Markov's inequality applied to the non-negative random variable $|S^{h}| - 1$ gives
\[
\mathbb{P}[H > h] \;=\; \mathbb{P}[|S^{h}| \geq 2] \;<\; \frac{4n}{\mu(h+1)}.
\]
Setting $h = \lceil 2(w+1)\log_{8/7}(4n) + 2 \rceil$ yields $\mu(h+1) \geq (4n)^{w+1}$, hence $\mathbb{P}[H > h] < 1/n^{w}$.
\end{proof}

%%%%%%%%%%%%%%%%%%%%%%%%%%%%%%%%%%%%%%%%%%%%%%%%%%%%%%%%%%%%%%%%%%%%%%
\subsection{Construction Algorithm}\label{app:restricted_recompression_construction}

Restricted recompression proceeds as follows.
At level $0$, for each position $s \in [1, n]$, the algorithm 
introduces or reuses a nonterminal $X_{i_{s}}$ with rule $X_{i_{s}} \rightarrow T[s]$ and determines $\assign(X_{i_{s}})$. 
%introduces a fresh nonterminal $X_{i_{s}}$ with rule $X_{i_{s}} \rightarrow T[s]$ and determines $\assign(X_{i_{s}})$. 
The sequence $S^{0}$ is then $X_{i_{1}}, X_{i_{2}}, \ldots, X_{i_{n}}$. 
Here, for two characters $T[s]$ and $T[s']$, $X_{i_{s}} = X_{i_{s'}}$ if and only if $T[s] = T[s']$.

For each $h \geq 1$, $S^{h}$ is built from $S^{h-1}$ in two steps. In the first step, $S^{h-1}$ is decomposed into segments $F_{1}, F_{2}, \ldots$ left-to-right as follows. Suppose $S^{h-1}[1..s-1]$ has been decomposed into $F_{1}, \ldots, F_{j-1}$; then the next segment $F_{j}$, starting at position $s$, is:
\begin{enumerate}
    \item $F_{j} = S^{h-1}[s..s+1]$ if $h$ is even, $\assign(S^{h-1}[s]) = 1$, and $\assign(S^{h-1}[s+1]) = 0$;
    \item $F_{j} = S^{h-1}[s..s+\ell-1]$, the longest run of $S^{h-1}[s]$ starting at position $s$, if $h$ is odd and $\assign(S^{h-1}[s]) \neq -1$;
    \item $F_{j} = S^{h-1}[s]$ otherwise.
\end{enumerate}

In the second step, for each segment $F_{j}$ of $S^{h-1}$, the algorithm introduces or 
reuses a nonterminal $X_{i_{j}}$ with rule $X_{i_{j}} \rightarrow F_{j}$ and determines $\assign(X_{i_{j}})$. Concatenating these nonterminals in order yields $S^{h}$. 
Here, for two segments $F_{j}$ and $F_{j'}$, $X_{i_{j}} = X_{i_{j'}}$ if and only if $F_{j} = F_{j'}$.

These two steps are applied for $h = 1, 2, \ldots$ until the sequence has length one; the final sequence is $S^{H}$, with $H \geq 1$ since $|S^{0}| = n \geq 2$.

The algorithm outputs an RLSLP $\mathcal{G}^{R} = (\mathcal{V}, \Sigma, \mathcal{D}, E)$, where $\mathcal{V}$ collects the nonterminals introduced across $S^{0}, S^{1}, \ldots, S^{H}$; $\Sigma$ is the alphabet of $T$; $\mathcal{D}$ is the set of production rules created above; and the start symbol $E$ is the unique nonterminal in $S^{H}$.

By Corollary~V.11 of~\cite{9961143}, the algorithm terminates with probability $1$, and its expected running time is $\mathcal{O}(n)$.

%% file: tex/9b_appendix.tex
\section{Details of the Dynamic RR-Index and Locate Algorithm}
\subparagraph{Representation of nodes and nonterminals.}
Each explicit node $u \in \mathcal{U}_{\expl}$ is stored as a pointer to the list element of $\pathset$ for the path containing $u$;
each implicit node is represented by the explicit node at the end of its path (i.e., the unique explicit endpoint; see Section~\ref{sec:component}).
In both cases, the list element is accessible in $\mathcal{O}(1)$ time.

Each nonterminal in $\mathcal{V}$ is represented by its DAG node.
Equality of two nonterminals is therefore tested by comparing DAG nodes; the dynamic RR-index never stores nonterminal subscripts explicitly.

%%%%%%%%%%%%%%%%%%%%%%%%%%%%%%%%%%%%%%%%%%%%%%%%%%%%%%%%%%%%%%%%%%%%%%
\subsection{Definition and Properties of the Popped Sequence}\label{app:popped_sequence}
\subparagraph{Definition.}
Let $\mathcal{G}_{P} = (\mathcal{V}_{P}, \Sigma_{P}, \mathcal{D}_{P}, E_{P})$ be the RLSLP obtained by applying restricted recompression to $P$.
Its production rules are shared with those of $\mathcal{G}^{R}$ whenever possible: if $X_{i} \rightarrow \expr_{i} \in \mathcal{D}_{P}$ and $\mathcal{D}$ contains a rule $X_{j} \rightarrow \expr_{j}$ with $\expr_{j} = \expr_{i}$, then $i = j$.
The derivation tree of $\mathcal{G}_{P}$ decomposes into $(\hat{H}+1)$ sequences
$S_{P}^{0}, S_{P}^{1}, \ldots, S_{P}^{\hat{H}}$.

The popped sequence of $P$ is defined using $3(\hat{H}+1)$ sequences of nonterminals
$L^{0}, L^{1}, \ldots, L^{\hat{H}}$, $R^{0}, R^{1}, \ldots, R^{\hat{H}}$,
and $C^{0}, C^{1}, \ldots, C^{\hat{H}}$. 
For each $h = 0, 1, \ldots, \hat{H}$, 
let $F_{1}, F_{2}, \ldots, F_{d}$ be the segments of $S_{P}^{h}$ obtained by restricted recompression. 
Suppose that $C^{h}$ is the sequence of nonterminals obtained by concatenating
$k \geq 0$ consecutive segments $F_{x}, F_{x+1}, \ldots, F_{x+k-1}$, where $x$ is an integer. 

\begin{description}
    \item[$L^{h}$:] $L^{h} = F_{x}$ unless $k = 0$ or $F_{x}$ consists of two distinct nonterminals. 
    In these two cases, $L^{h}$ is the empty sequence. 
    \item[$R^{h}$:] $R^{h} = F_{x+k-1}$ unless $k \leq 1$ or $F_{x+k-1}$ consists of two distinct nonterminals. 
    In these two cases, $R^{h}$ is the empty sequence. 
    \item[$C^{h}$:] $C^{0} = S_{P}^{0}$.
    For $h \geq 1$, let $\bar{C}^{h-1}$ be the sequence of nonterminals with $C^{h-1} = L^{h-1} \cdot \bar{C}^{h-1} \cdot R^{h-1}$.
    Then $\bar{C}^{h-1}$ corresponds to $k' \geq 0$ consecutive segments $F_{x'}, F_{x'+1}, \ldots, F_{x'+k'-1}$ for some integer $x'$,
    and for each $j \in \{1, 2, \ldots, k'\}$ the nonterminals in $F_{x' + j - 1}$ share the same parent $X_{i_{j}}$ in the derivation tree.
    Set $C^{h} = X_{i_{1}}, X_{i_{2}}, \ldots, X_{i_{k'}}$.
\end{description}
Let $q \geq 0$ be the largest integer satisfying $|C^{q}| \geq 1$. 
Then the popped sequence $Q$ of $P$ is defined as the concatenation of
$L^{0}, L^{1}, \ldots, L^{q}, R^{q}, R^{q-1}, \ldots, R^{0}$. 

\subparagraph{Properties.}
The popped sequence $Q$ has four properties.
First, the number of runs in the run-length encoding of $Q$ is $\mathcal{O}(q)$, since 
each $L^{h}$ and each $R^{h}$ contains at most one run. 
%each $L^{h}$ and each $R^{h}$ has run-length at most $1$.

Second, the number of runs in the run-length encoding of $Q$ is $\mathcal{O}(\log m)$ in expectation. 
This is because $q \leq \hat{H}$, 
and Lemma~\ref{lem:tree_height} gives $\mathbb{E}[\hat{H}] = \mathcal{O}(\log m)$. 

Third, if $P$ occurs in $T$, then the nodes corresponding to each $L^{h}$ (resp.\ $R^{h}$) share a common parent in the derivation tree of $\mathcal{G}^{R}$, because each $L^{h}$ (resp.\ $R^{h}$) is a segment that restricted recompression replaces by a single nonterminal. 

Fourth, for $m \geq 2$, the popped sequence $Q$ contains at least two nonterminals (i.e., $|Q| \geq 2$).
Indeed, $C^0=S_P^0$ has length $m$. Let $F_1,\ldots,F_k$ be the segments of
$S_P^0$. Since height $0$ is even, no segment consists of two distinct
nonterminals. If $k=1$, then $L^0=F_1$ has length $m\geq 2$. If $k\geq 2$, then
both $L^0=F_1$ and $R^0=F_k$ are nonempty. Hence $|Q|\geq 2$ in either case.

%%%%%%%%%%%%%%%%%%%%%%%%%%%%%%%%%%%%%%%%%%%%%%%%%%%%%%%%%%%%%%%%%%%%%%
\subsection{Proof of Lemma~\ref{lem:core_rectangle}}\label{app:core_rectangle}
This lemma generalizes Lemma~7 in \cite{DBLP:journals/dam/NishimotoIIBT20}. 
Similar lemmas also appear in 
\cite{DBLP:conf/latin/ChristiansenE18,ViceVersa,DBLP:journals/talg/ChristiansenEKN21,DBLP:journals/dam/NishimotoIIBT20}. 
In \cite{DBLP:journals/dam/NishimotoIIBT20}, Lemma~7 is proved for the RLSLP $\mathcal{G}^{R}$ constructed by signature encoding, using the following property.
\begin{description}
    \item[Property.]
    Let $(p', h')$ be a pair of integers such that $S^{h'}[p'] = S^{h'}[p'+1]$. 
    Let $v_{p'}$ and $v_{p'+1}$ be the nodes in the derivation tree of $\mathcal{G}^{R}$ corresponding to
    the two nonterminals $S^{h'}[p']$ and $S^{h'}[p'+1]$, respectively. 
    Then all children of the lowest common ancestor of $v_{p'}$ and $v_{p'+1}$ have the same label $X \in \mathcal{V}$, 
    and $\val(X) = \val(S^{h'}[p'])$. 
\end{description}

This property is common to RLSLPs constructed by \emph{locally consistent parsing algorithms}, 
such as signature encoding, recompression, signature grammar, and restricted block compression. 
Since restricted recompression is also such a parsing algorithm, 
Lemma~\ref{lem:core_rectangle} can be proved using the same property.

%We now give the detailed proof, by contradiction. 
We now give a detailed proof by contradiction. 
Suppose that some integer $s \in \{1,2,\ldots,m-1\}$ is such that
$\rect_{s}$ contains a point of $\mathcal{P}$ while
$s \notin \{ |\val(Q[1])| \} \cup \{ \phi(1), \phi(2), \ldots, \phi(\rho-1) \}$.

If no integer $s' \in \{1,2,\ldots,|Q|-1\}$ satisfies
$s = |\val(Q[1..s'])|$, then the split position $s$ lies strictly inside the
string derived by a single nonterminal of $Q$. In that case, the first boundary
of the production rule of the explicit node corresponding to a point in
$\rect_s$ cannot coincide with this split, contradicting the definition of a
primary occurrence. Hence such an $s'$ exists. 
Fix such an integer $s'$. 
Since $s \notin \{ \phi(1), \phi(2), \ldots, \phi(\rho-1) \}$, the boundary at position $s$ does not separate two distinct runs of $Q$, and hence $Q[s'] = Q[s'+1]$. 

%Fix such an $s'$ and set $P_{L} = \val(Q[1..s'])$.
%By the definition of the core, there exists an integer $s' \in \{1,2,\ldots,|Q|-1\}$ with $|P_{L}| = s$, where $P_{L} = \val(Q[1..s'])$.
%If no such $s'$ exists, then the split position $s$ lies strictly inside the string derived by a single nonterminal of $Q$. In that case, the first boundary of the production rule of the explicit node corresponding to a point in $\rect_s$ cannot coincide with this split, contradicting the definition of a primary occurrence. Hence such an $s'$ exists.

Let $X_i \in \mathcal{V}$ be the nonterminal corresponding to a point $t$ in $\rect_{s}$. 
The nonterminal $X_i$ occurs in $S^{h}$ at position $p \in \{ 1, 2, \ldots, |S^{h}| \}$, 
where $h$ is the height of $X_i$, 
and $X_i$ produces a sequence of $d$ nonterminals
$X_{i_1}, X_{i_2}, \ldots, X_{i_d}$. 
The $x$-coordinate $L_i$ of $t$ is the reverse of $\val(X_{i_1})$. 
Similarly, 
$Q[s'..s'+1]$ occurs in $S^{h'}$ at position $p' \in \{ 1, 2, \ldots, |S^{h'}|-1 \}$, 
where $h'$ is the height of $Q[s']$, 
and $h' < h$.

Let $v_p$, $v_{p'}$, and $v_{p'+1}$ be the nodes in the derivation tree
corresponding to $S^{h}[p]$, $S^{h'}[p']$, and $S^{h'}[p'+1]$, respectively. 
Let $v_{\mathsf{lca}}$ be the lowest common ancestor of $v_{p'}$ and $v_{p'+1}$. 
By the common property of locally consistent parsing algorithms, 
all children of $v_{\mathsf{lca}}$ have the same label $X \in \mathcal{V}$, 
and $\val(X) = \val(S^{h'}[p']) = \val(Q[s'])$. 
Therefore, $v_{p'}$ does not occur on the path from $v_{\mathsf{lca}}$ to 
the rightmost leaf in the subtree rooted at $v_{\mathsf{lca}}$. 
%the rightmost leaf under the subtree rooted at $v_{\mathsf{lca}}$. 

We next show that $v_p = v_{\mathsf{lca}}$.
If instead $v_{\mathsf{lca}}$ were a proper descendant of $v_p$, then $v_{p'}$ would not occur on the path from $v_p$ to the rightmost leaf of the subtree rooted at the leftmost child of $v_p$; but $(X_i, s)$ being a primary occurrence forces $v_{p'}$ onto this path, a contradiction.

Since $L_i$ is the reverse of $\val(X_{i_1})$ and $\val(X_{i_1}) = \val(X) = \val(Q[s'])$, we have $L_i = \reverse(\val(Q[s']))$, which forces $s' = 1$. 
Since $s'=1$, the equality $s = |\val(Q[1..s'])|$ gives
$s = |\val(Q[1])|$, contradicting the standing assumption that
$s \neq |\val(Q[1])|$.

%Then $P_L = \val(Q[1])$ gives $s = |\val(Q[1])|$, contradicting the standing assumption that $s \neq |\val(Q[1])|$.

%%%%%%%%%%%%%%%%%%%%%%%%%%%%%%%%%%%%%%%%%%%%%%%%%%%%%%%%%%%%%%%%%%%%%%
\subsection{Proof of Lemma~\ref{lem:basic_queries}}\label{app:basic_queries}
We first establish how edge labels are recovered from the production rules in the doubly linked list.

\begin{observation}\label{ob:comp_edges}
Let $e \in \mathcal{E}$ be a directed edge.
If $e$ is an intra-path edge, then $\mathcal{L}_{E}(e) = 0$.
Otherwise (i.e., $e$ is an inter-path edge in $\mathcal{E}_{\expl}(\mathbb{P})$ for a path $\mathbb{P} \in \pathset$),
$e$ corresponds to an occurrence of $X_{j}$ in the right-hand side of the production rule $X_{i} \rightarrow X_{i_{1}} X_{i_{2}} \cdots X_{i_{d}} \in \mathcal{D}$, where $X_{i} = \mathcal{L}_{U}(\mathrm{src}(e))$ and $X_{j} = \mathcal{L}_{U}(\mathrm{dst}(e))$.
In this case, $\mathcal{L}_{E}(e) = \sum_{s=1}^{k-1} |\val(X_{i_{s}})|$, where $k$ is the position of the occurrence of $X_{j}$ in the right-hand side.
\end{observation}

By Observation~\ref{ob:comp_edges}, every edge label can be reconstructed in $\mathcal{O}(1)$ time from the production rule at the source.
Moreover, each value in $\vOcc(u_i)$ corresponds to a path from the root to $u_i$ in the DAG and equals $1$ plus the sum of edge labels along that path.

\begin{proof}[Proof of Lemma~\ref{lem:basic_queries}]
Let $u_j \in \mathcal{U}_{\expl}$ be the explicit node on the path $\mathbb{P}$ (i.e., the last node of the path),
and let $X_j \in \mathcal{V}$ be the corresponding nonterminal.

\subparagraph{Computation of $\expr_i$.}
If $u_i$ is explicit, then $\expr_i$ is stored in the list element corresponding to $\mathbb{P}$. 
Otherwise, $u_i$ has a child $u_k$ on $\mathbb{P}$, and $\expr_i = X_k$, where $X_k$ is the nonterminal corresponding to $u_k$. 
Therefore, $\expr_i$ can be computed in $\mathcal{O}(1)$ time.

\subparagraph{Computation of $|\val(X_i)|$.}
We have $|\val(X_i)| = |\val(X_j)|$, and $|\val(X_j)|$ is stored in the list element corresponding to $\mathbb{P}$. 
Therefore, $|\val(X_i)|$ can be computed in $\mathcal{O}(1)$ time.

\subparagraph{Computation of the height $h$ of $u_{i}$.}
If $u_i$ is explicit, then its height $h$ is stored in the list element corresponding to $\mathbb{P}$. 
Otherwise, the height $h'$ of $u_j$ is stored in the list element corresponding to $\mathbb{P}$, and $h = h' + s$, where $s$ is the distance from $u_i$ to $u_j$ along $\mathbb{P}$. 
Since $s$ can be computed in $\mathcal{O}(1)$ time, $h$ can also be computed in $\mathcal{O}(1)$ time.

\subparagraph{Computation of $\assign(X_i)$.}
$\assign(X_i)$ is stored in the list element corresponding to $\mathbb{P}$, and hence $\assign(X_i)$ can be computed in $\mathcal{O}(1)$ time.

\subparagraph{Computation of the distinct children of $u_i$.}
The distinct children of $u_i$ are exactly the nodes whose corresponding nonterminals appear in $\expr_i$. 
Since the run-length encoding of $\expr_i$ has length at most $2$, 
the distinct children of $u_i$ can be computed in $\mathcal{O}(1)$ time.
%Since the run-length of $\expr_i$ is at most $2$, 
%the children of $u_i$ can be computed in $\mathcal{O}(1)$ time.

\subparagraph{Proof of Lemma~\ref{lem:basic_queries}(ii): outgoing edges.}
If $\expr_i$ is a character, then $u_i$ has no outgoing edges.
Otherwise, $\expr_i$ is a sequence of nonterminals $X_{i_1}, X_{i_2}, \ldots, X_{i_d}$, and for each nonterminal $X_{i_s}$, $u_i$ has an outgoing edge $e_s$ to the node $u_{i_s}$ corresponding to $X_{i_s}$.
The label $\mathcal{L}_{E}(e_s)$ of $e_s$ is $\sum_{t=1}^{s-1} |\val(X_{i_t})|$.
By Lemma~\ref{lem:basic_queries}(i), we can compute $\expr_i$ and the $d-1$ integers $|\val(X_{i_1})|, |\val(X_{i_2})|, \ldots, |\val(X_{i_{d-1}})|$.
Therefore, the outgoing edges of $u_i$ can be computed in $\mathcal{O}(1)$ time per edge.

\subparagraph{Proof of Lemma~\ref{lem:basic_queries}(iii): inter-path edges.}
For each key-value pair $(\mathrm{src}(e), \mathrm{dst}(e))$ stored in the per-path hash table for $\mathcal{E}_{\expl}(\mathbb{P})$, 
we compute the inter-path edges from $\mathrm{src}(e)$ to $\mathrm{dst}(e)$ using Lemma~\ref{lem:basic_queries}(ii). 
In this way, we obtain all edges in $\mathcal{E}_{\expl}(\mathbb{P})$. 
This computation takes $\mathcal{O}(1)$ time per edge.

\subparagraph{Proof of Lemma~\ref{lem:basic_queries}(iv): enumerating $\vOcc(u_i)$.}
Since the label of each edge on $\mathbb{P}$ is $0$, 
the following equation follows from the relationship between $\vOcc(u_i)$ and $\mathcal{L}_{E}$: 
\[
\vOcc(u_i) =
\bigcup_{e \in \mathcal{E}_{\expl}(\mathbb{P})}
\{\, \mathcal{L}_{E}(e) + q \mid q \in \vOcc(\mathrm{src}(e)) \,\}
\]
if $|\mathcal{E}_{\expl}(\mathbb{P})| \neq 0$, and $\vOcc(u_i) = \{ 1 \}$ otherwise.

Hence $\vOcc(u_i)$ can be computed recursively.
By Lemma~\ref{lem:basic_queries}(iii), the edges in $\mathcal{E}_{\expl}(\mathbb{P})$ and their labels are enumerated in $\mathcal{O}(1)$ time per edge, so the recursion runs in $\mathcal{O}(\tau)$ time, where $\tau$ is the total edge count over all root-to-$u_i$ paths in the DAG.
Since the DAG has height $H$ and contains $|\vOcc(u_i)|$ such paths, $\tau \leq H |\vOcc(u_i)|$, proving Lemma~\ref{lem:basic_queries}(iv).
\end{proof}

%%%%%%%%%%%%%%%%%%%%%%%%%%%%%%%%%%%%%%%%%%%%%%%%%%%%%%%%%%%%%%%%%%%%%%
\subsection{Details of Random Access, LCE, and Reversed LCE Queries}\label{app:lce}
\subparagraph{Random access queries.}
We compute $T[i..i+\ell-1]$ by traversing the derivation tree of $\mathcal{G}^{R}$ from the root to
the $\ell$ leaves corresponding to the substring. 
Using Lemma~\ref{lem:basic_queries}, 
this traversal can be performed in $\mathcal{O}(\kappa)$ time, 
where $\kappa$ is the number of nodes in the derivation tree that satisfy the following two conditions:
(i) the label of the node corresponds to an explicit node in the DAG;
(ii) the node is an ancestor of a leaf corresponding to a character of $T[i..i+\ell-1]$.
$\kappa$ can be bounded by $\mathcal{O}(H + \ell)$.
Therefore, $T[i..i+\ell-1]$ can be computed in $\mathcal{O}(H + \ell)$ time.

\subparagraph{LCE and reversed LCE queries.}
Consider an LCE query on two suffixes $T[i..n]$ and $T[j..n]$, and let $\ell \geq 0$ denote the length of their longest common prefix. 
If $\ell=0$, the query is answered immediately. Otherwise, 
let $Q$ be the shortest sequence of nonterminals that (i) derives $T[i..i+\ell-1]$ and (ii) is embedded at both positions $i$ and $j$.
The standard algorithm finds $Q$ by traversing the derivation tree and recovers $\ell$ as the sum of $|\val(\cdot)|$ over the nonterminals of $Q$; reversed LCE is symmetric.

Duyster and Kociumaka~\cite{DBLP:journals/mst/DuysterK26} showed that this algorithm runs in $\mathcal{O}(H)$ time on the restricted recompression RLSLP, and Lemma~\ref{lem:basic_queries} lets us simulate it within the same bound.
Given the popped sequence of $P$, the same procedure extends to an LCE between a suffix of $T$ and a suffix of $P$ (traverse the subtrees rooted at the nodes of the popped sequence), and likewise to a reversed LCE between prefixes.

\subparagraph{Comparing a suffix $P_R$ of $P$ with the coordinate string $R_i$ of an element of $\mathcal{Y}$.}
We compare $P_{R}$ and $R_{i}$ in the following six steps.
{\renewcommand{\labelenumi}{(\roman{enumi})}
\begin{enumerate}
\item Access the node corresponding to $R_{i}$ using the list indexing data structure built on $\mathcal{Y}$. 
\item Compute the length $\ell$ of the longest common prefix of $P_{R}$ and $R_{i}$ using an LCE query.
\item If $\ell = |P_{R}|$ and $|P_{R}| < |R_{i}|$, then $P_{R}$ is lexicographically smaller than $R_{i}$.
\item If $\ell = |R_{i}|$ and $|P_{R}| > |R_{i}|$, then $R_{i}$ is lexicographically smaller than $P_{R}$.
\item If $\ell = |P_{R}| = |R_{i}|$, then $P_{R}$ and $R_{i}$ are equal.
\item Otherwise, comparing $P_{R}$ and $R_{i}$ reduces to comparing $P_{R}[\ell+1]$ and $R_{i}[\ell+1]$. 
We access $R_{i}[\ell+1]$ by a random access query and return the result of this character comparison.
\end{enumerate}
}
This algorithm runs in $\mathcal{O}(H + \log M)$ time.

\subparagraph{Comparing the reverse of a prefix $P_L$ of $P$ with the coordinate string $L_i$ of an element of $\mathcal{X}$.}
This comparison can be performed in $\mathcal{O}(H + \log M)$ time
by modifying the algorithm for comparing $P_{R}$ and $R_{i}$.

%%%%%%%%%%%%%%%%%%%%%%%%%%%%%%%%%%%%%%%%%%%%%%%%%%%%%%%%%%%%%%%%%%%%%%
\subsection{Algorithm for Constructing the Popped Sequence}\label{app:pop}
To explain how to construct the popped sequence, 
we introduce a \emph{rule query}: given a right-hand side $\expr$, return the nonterminal $X \in \mathcal{V}$ such that $X \rightarrow \expr \in \mathcal{D}$, or $\bot$ if no such nonterminal exists. 
This query is based on the following observation.
If $X$ corresponds to an explicit node, then it is stored in the global hash table as the value associated with key $\expr$.
Otherwise, $X$ corresponds to an implicit node $u \in \mathcal{U}$; in this case, $\expr$ is a single nonterminal $X' \in \mathcal{V}$, and the node corresponding to $X'$ is the child of $u$ on a path in $\pathset$. 
Using this observation, 
a rule query can be answered in $\mathcal{O}(1)$ time using the data structures maintained by the dynamic RR-index.

The popped sequence of $P$ is built bottom up using rule queries, as in restricted recompression.
If the construction would require creating a new nonterminal, the algorithm aborts. 
In this case, $\Occ(T, P) = \emptyset$ follows from the definition of the core. 
Otherwise it runs in $\mathcal{O}(\sum_{h=0}^{\hat{H}} |S^{h}_{P}|)$ time.
Kociumaka showed $\mathbb{E}[\sum_{h=0}^{\hat{H}} |S^{h}_{P}|] = \mathcal{O}(m)$ (Lemma~V.10 in \cite{9961143}), so the expected running time is $\mathcal{O}(m)$.

%%%%%%%%%%%%%%%%%%%%%%%%%%%%%%%%%%%%%%%%%%%%%%%%%%%%%%%%%%%%%%%%%%%%%%
\subsection{Detailed Proof of \texorpdfstring{$Z = \mathcal{O}(|\vOcc(u_i)| \lceil H / \log B \rceil)$}{Z}}\label{app:z_bound}
\begin{figure}[t]
 \begin{center}
		\includegraphics[scale=0.6]{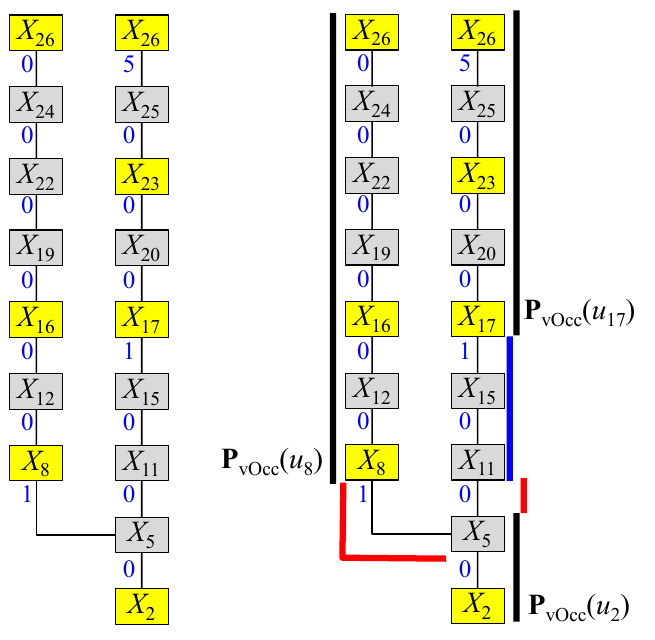}

\caption{
(Left) The tree $\mathcal{T}_{\vOcc}(u_2)$, where $u_2$ is a node in the DAG shown in Figure~\ref{fig:rrdag}. 
(Right) The tree decomposition of $\mathcal{T}_{\vOcc}(u_2)$ proposed in Appendix~\ref{app:z_bound}. 
The bold red, black, and blue lines represent the paths corresponding to Cases 1, 2, and 3, respectively.
}
\label{fig:vOccTree}
 \end{center}
\end{figure}
To prove $Z = \mathcal{O}(|\vOcc(u_i)| \lceil H / \log B \rceil)$,
we introduce a labeled tree $\mathcal{T}_{\vOcc}(u_i)$ that represents all directed paths from the root of the DAG to $u_i$. 
The root is labeled with the nonterminal $X_i$ corresponding to $u_i$. 
For each directed edge $e \in \mathcal{E}$ from a node $u_{j} \in \mathcal{U}$ to $u_i$, 
the root has a child labeled with the nonterminal $X_j$ corresponding to $u_j$, 
and the edge between the root and the child is labeled with $\mathcal{L}_{E}(e)$. 
The subtree rooted at that child is recursively defined as $\mathcal{T}_{\vOcc}(u_j)$. 
Thus, each root-to-leaf path in $\mathcal{T}_{\vOcc}(u_i)$ corresponds, in reverse order, to a directed path from the root of the DAG to $u_i$.

We recursively decompose $\mathcal{T}_{\vOcc}(u_i)$ into paths according to the following four cases:
\begin{description}
    \item[Case 1.]
    The root has at least two children. 
    In this case, the edge between the root and each child forms a path of length $1$. 
    If a child is not a leaf, then the subtree rooted at that child is recursively decomposed into paths.

    \item[Case 2.]
    The root has a single child, and $u_i$ is explicit. 
    In this case, the tree is decomposed into 
    (i) the path $\mathbb{P}$ in $\mathcal{T}_{\vOcc}(u_i)$ corresponding to $\mathbb{P}_{\vOcc}(u_i)$, and 
    (ii) the subtree rooted at the last node of $\mathbb{P}$. 
    If the last node of $\mathbb{P}$ is not a leaf, then the subtree rooted at that node is recursively decomposed into paths.

    \item[Case 3.]
    The root has a single child, and $u_i$ is implicit.
    In this case, $u_i$ lies on a path $\mathbb{P}$ 
    in $\pathset$ whose first node corresponds to a nonterminal $X_j$, 
    and in the tree $\mathcal{T}_{\vOcc}(u_i)$, 
    there exists a path from the root to the lowest node $v$ satisfying one of the following two conditions: 
    (A) the node has at least two children; 
    (B) an ancestor of the node is labeled with $X_{j}$.    
    %(B) the ancestor of the node is labeled with $X_{j}$.    
    The tree is decomposed into
    (i) $\mathbb{P}$, and
    (ii) the subtree rooted at $v$.
    If $v$ is not a leaf, the subtree is recursively decomposed into paths.

    \item[Case 4.]
    The root has no child. 
    In this case, the tree is decomposed into a path of length $0$.
\end{description}

Figure~\ref{fig:vOccTree} illustrates $\mathcal{T}_{\vOcc}(u_2)$ and its tree decomposition for an explicit node $u_2 \in \mathcal{U}$. The tree consists of sixteen nodes, and the decomposition yields six paths. 

Let $Z'$ be the number of paths obtained by the tree decomposition. 
The following lemma gives an upper bound on $Z'$. 

\begin{lemma}
$Z' = \mathcal{O}(|\vOcc(u_{i})| \lceil H / \log B \rceil)$.
\end{lemma}
\begin{proof}
Remove from $\mathcal{T}_{\vOcc}(u_i)$ every edge whose upper endpoint has at least two children.
This partitions the tree into $\mathcal{O}(|\vOcc(u_i)|)$ paths.
Each resulting path $\mathbb{P}$ is further decomposed into $\mathcal{O}(\lceil |\mathbb{P}| / \log B \rceil)$ subpaths by the tree decomposition; since $|\mathbb{P}| \leq H$, the total subpath count is $\mathcal{O}(|\vOcc(u_i)| \lceil H / \log B \rceil)$.
\end{proof}

The recursive algorithm for $\vOcc(u_{i})$ mirrors the tree decomposition of $\mathcal{T}_{\vOcc}(u_i)$, so $Z = \mathcal{O}(Z')$, and $Z = \mathcal{O}(|\vOcc(u_i)| \lceil H / \log B \rceil)$ follows.

%%%%%%%%%%%%%%%%%%%%%%%%%%%%%%%%%%%%%%%%%%%%%%%%%%%%%%%%%%%%%%%%%%%%%%
\subsection{Speeding up Phase (ii) in Practice}\label{app:practical_phase2}
The bottleneck in phase (ii) is comparing the $x$- or $y$-coordinate $s' \in \Sigma^{*}$ of a point in $\mathcal{P}$ with a substring $s$ of $P$ (either a suffix of $P$ or the reverse of a prefix of $P$) using LCE, reversed LCE, and random access queries (see also Section~\ref{subsubsec:aux_queries}).
To reduce the number of calls to these three queries, we use the following observation. 
Let $\lambda = \lceil \beta B / \log \sigma \rceil$.
Let $s_{\short}$ and $s'_{\short}$ be the prefixes of $s$ and $s'$ of lengths
$\min\{|s|,\lambda\}$ and $\min\{|s'|,\lambda\}$, respectively. 
%Let $s_{\short} \in \Sigma^{*}$ and $s'_{\short} \in \Sigma^{*}$ be the prefixes of length $\lceil \beta B / \log \sigma \rceil$ of $s$ and $s'$, respectively.
If $s_{\short} \neq s'_{\short}$, then the lexicographic order between $s$ and $s'$ is the same as that between $s_{\short}$ and $s'_{\short}$.
Therefore, when the longest common prefix of $s$ and $s'$ is sufficiently short, we can determine the lexicographic order between $s$ and $s'$ without using these queries.
In practice, we expect this case to occur frequently.

We compare $s$ and $s'$ in the following three steps:
(i) compute $s_{\short}$ and access $s'_{\short}$ stored in the doubly linked list for $\pathset$;
(ii) compare $s_{\short}$ and $s'_{\short}$;
(iii) if $s_{\short} \neq s'_{\short}$, return this comparison as the result for $s$ and $s'$; otherwise, fall back to the algorithm of Section~\ref{subsubsec:aux_queries}.

%% file: tex/9c_appendix.tex
\section{Details of the Insertion Operation}\label{app:insertion_operation}
\subsection{Details of Phase 1}\label{app:insertion_phase1}
\subparagraph{Updating the dynamic RR-index.}
Every time a new production rule $X_{i} \rightarrow \expr_{i}$ is established, 
a new node $u_{i}$ of height $h$ corresponding to $X_{i}$ is inserted into the DAG. 
If $u_{i}$ is explicit, then a new path $\mathbb{P} = u_{i}$ is inserted into $\pathset$, 
and a new directed edge $e$ from $u_{i}$ to the node $u_{j}$ corresponding to each nonterminal $X_{j}$ in $\expr_{i}$ 
is inserted into $\mathcal{E}_{\expl}(\mathbb{P}_{j})$, 
where $\mathbb{P}_{j} \in \pathset$ is the path containing $u_{j}$.

Otherwise (i.e., $u_{i}$ is implicit), 
$\expr_{i}$ consists of a single nonterminal $X_{j'}$, 
the corresponding node $u_{j'}$ is the first node of a path $\mathbb{P}' \in \pathset$, 
and the path is extended by connecting $u_{i}$ and $\mathbb{P}'$. 
To reflect the changes to $\pathset$ and $\mathcal{E}_{\expl}(\mathbb{P}_{j})$, 
we need to update the following data structures: 

\begin{description}
    \item[The global hash table.]
    If $u_{i}$ is explicit, 
    %$X_{i}$ and the binary representation of $\expr_{i}$ are inserted into the hash table as a value and its key, respectively. 
    the binary representation of $\expr_{i}$ is inserted as the key and $X_{i}$ as the value.
    This insertion takes expected amortized $\mathcal{O}(1)$ time~\cite{DBLP:journals/siamcomp/DietzfelbingerKMHRT94}. 
    
    \item[The list element corresponding to $\mathbb{P}$.]
    If $u_{i}$ is explicit,
    then create a new element corresponding to $\mathbb{P}$,
    append it to the doubly linked list for $\pathset$,
    and associate an empty hash table with it.
    This element stores the following five pieces of information: 
    (i) a pair $(u_{i}, 0)$ for recovering $\mathbb{P}$ from the element, 
    (ii) $\expr_{i}$, (iii) $h$, (iv) $|\val(X_{i})|$, 
    and (v) $\assign(X_{i})$. 
    Here, $|\val(X_{i})|$ is computed by summing the lengths of 
    the substrings derived from the nonterminals $X_{j}$ in $\expr_{i}$ 
    %the substring derived from each nonterminal $X_{j}$ in $\expr_{i}$ 
    if $\expr_{i}$ is a sequence of nonterminals. Otherwise, $|\val(X_{i})| = 1$. 
    $|\val(X_{i})|$ can be computed in $\mathcal{O}(1)$ time using Lemma~\ref{lem:basic_queries}. 
    Therefore, the list element can be updated in $\mathcal{O}(1)$ time. 
    
    \item[The list element corresponding to $\mathbb{P'}$.]
    If $u_{i}$ is implicit,
    then the second component of the pair stored in the list element is increased by one,
    and $\assign(X_{i})$ is appended to the list element.
    These updates take $\mathcal{O}(\max \{1, h/B \})$ time.
    
    \item [The per-path hash table corresponding to $\mathbb{P}_{j}$.]
    The key and value of the directed edge $e$ are inserted into the hash table 
    if the hash table does not already contain the key. 
    This insertion takes expected amortized $\mathcal{O}(1)$ time. 
    In addition, the key of $e$ is inserted into the doubly linked list for the keys at 
    an appropriate position. 
    %The list indexing data structure can support random access to the keys 
    %and an insertion into the list in $\mathcal{O}(\log M)$ time~\cite{DBLP:conf/wads/Dietz89}. 
    The list indexing data structure supports random access to the keys and insertions into the list 
    in $\mathcal{O}(\log M)$ time~\cite{DBLP:conf/wads/Dietz89}. 
    To compute the appropriate insertion position, 
    we need $\mathcal{O}(\log M)$ random accesses to the keys. 
    Therefore, the update of the doubly linked list takes $\mathcal{O}(\log^{2} M)$ time in total. 

    \item[The dynamic data structures for $\mathcal{X}$.]
    If $u_{i}$ is explicit,
    then $L_{i}$ is inserted after the largest element in $\mathcal{X}$ preceding the $x$-coordinate $L_{i}$ of $u_{i}$,
    and the list indexing and order maintenance data structures for $\mathcal{X}$ are updated accordingly in $\mathcal{O}(\log M)$ time.
    The largest element is found by a binary search on $\mathcal{X}$ in $\mathcal{O}(\max \{ H, H', \log M \} \log M)$ time, using the same approach as in the locate query algorithm.

    \item[The dynamic data structures for $\mathcal{Y}$.]
    As with the update of the dynamic data structures for $\mathcal{X}$, 
    $R_{i}$ is inserted after the largest element in $\mathcal{Y}$ preceding the $y$-coordinate $R_{i}$ of $u_{i}$, 
    and the list indexing data structure and order maintenance data structure corresponding to $\mathcal{Y}$ 
    are updated accordingly. 
    This update takes $\mathcal{O}(\max \{ H, H', \log M \} \log M)$ time. 
    
    \item[The dynamic data structures for $\mathcal{P}$.]    
    If $u_{i}$ is explicit,     
    then the corresponding point is inserted into $\mathcal{P}$,
    and the dynamic data structure for range reporting queries is updated accordingly in amortized $\mathcal{O}(\log M)$ time.
\end{description}

These updates take amortized $\mathcal{O}(\max \{ H, H', \log M \} \log M)$ time if $u_{i}$ is explicit; 
otherwise, they take $\mathcal{O}(\max \{1, h/B \})$ time, 
where $h / B = \mathcal{O}(1)$. 

\subparagraph{Algorithm.}
The first phase consists of the following two steps.
{\renewcommand{\labelenumi}{(\roman{enumi})}
\begin{enumerate}
\item Compute $Q_{L}$ and $Q_{R}$ using the algorithm of~\cite{DBLP:journals/mst/DuysterK26}. 
This is done by traversing the derivation tree of $\mathcal{G}^{R}$ from the root to 
the nodes corresponding to the sequence $Q_{L}$ embedded at position $1$ and the sequence $Q_{R}$ embedded at position $s$, respectively. 
By Lemma~\ref{lem:basic_queries}, 
these traversals can be performed in $\mathcal{O}(H)$ time.
\item Compute $\mathcal{T}_{\diff}$ in $\mathcal{O}(|\mathcal{T}_{\diff}|)$ time using the rule query introduced in Appendix~\ref{app:pop}. 
This computation may create new production rules. 
Whenever a new production rule $X_{i} \rightarrow \expr_{i}$ is established, 
we update the dynamic RR-index accordingly. 
\end{enumerate}
}

Therefore, the expected amortized running time of the first phase is  
\begin{equation*}
\mathcal{O}(|\mathcal{T}_{\diff}| + |\mathcal{I}_{\expl}| \max \{ H, H', \log M \} \log M + |\mathcal{I}_{\impl}|),
\end{equation*}
where $\mathcal{I}_{\expl}$ and $\mathcal{I}_{\impl}$ are introduced in Section~\ref{sec:update}. 

%%%%%%%%%%%%%%%%%%%%%%%%%%%%%%%%%%%%%%%%%%%%%%%%%%%%%%%%%%%%%%%%%%%%%%
\subsection{Details of Phase 2}\label{app:insertion_phase2}
\subparagraph{Computation of $\mathcal{R}_{\expl} \cup \mathcal{R}_{\impl}$.}
To update the dynamic RR-index, 
we need to identify the nodes in $\mathcal{R}_{\expl} \cup \mathcal{R}_{\impl}$. 
For this purpose, we use the following observation:
a node $u_i$ belongs to $\mathcal{R}_{\expl} \cup \mathcal{R}_{\impl}$ 
if and only if all parents of $u_i$ belong to $\mathcal{R}_{\expl} \cup \mathcal{R}_{\impl}$.

We start from the node $u_i$ corresponding to the start symbol $E$ of $\mathcal{G}^{R}$, letting $\mathbb{P} \in \pathset$ denote the path containing $u_i$; this $u_i$ has an incoming edge only if $E$ still appears in $\mathcal{G}^{R'}$.

$\mathcal{R}_{\expl} \cup \mathcal{R}_{\impl}$ can be computed recursively as follows:
{\renewcommand{\labelenumi}{(\roman{enumi})}
\begin{enumerate}
\item Determine whether $u_i$ has an incoming edge. 
This can be verified in $\mathcal{O}(1)$ time:
$u_i$ has at least one incoming edge if
(A) $u_i$ is not the first node of $\mathbb{P}$, or
(B) $u_i$ is the first node of $\mathbb{P}$ and one incoming edge of $u_i$ is explicitly stored in the list element corresponding to $\mathbb{P}$.
\item If $u_i$ has an incoming edge, then the recursion stops.
Otherwise, $u_i$ belongs to $\mathcal{R}_{\expl} \cup \mathcal{R}_{\impl}$;
we output $u_i$, remove $u_i$ and its outgoing edges from the DAG, 
and recursively process the target node of each outgoing edge of $u_i$ that has not yet been removed.
\end{enumerate}
}
By the observation above, this recursive algorithm correctly computes
$\mathcal{R}_{\expl} \cup \mathcal{R}_{\impl}$.
Since it processes $\mathcal{O}(|\mathcal{R}_{\expl}| + |\mathcal{R}_{\impl}|)$ nodes,
its running time is $\mathcal{O}(|\mathcal{R}_{\expl}| + |\mathcal{R}_{\impl}|)$.
This bound excludes the update time of the dynamic RR-index.

\subparagraph{Updating the dynamic RR-index.}
Whenever a node is removed from the DAG, the dynamic RR-index is updated accordingly.
As in Phase~1, this update takes expected amortized $\mathcal{O}(\log M)$ time if the removed node is explicit, and $\mathcal{O}(1)$ time otherwise.
Therefore, Phase~2 takes expected amortized $\mathcal{O}(|\mathcal{R}_{\expl}| \log M + |\mathcal{R}_{\impl}|)$ time in total.

%%%%%%%%%%%%%%%%%%%%%%%%%%%%%%%%%%%%%%%%%%%%%%%%%%%%%%%%%%%%%%%%%%%%%%
\subsection{Details of Phase 3}\label{app:insertion_phase3}
To compute $\mathcal{U}_{\change}$, 
we introduce a set $\mathcal{U}_{\children} \subseteq \mathcal{U}$ of nodes $u$ satisfying the following two conditions: 
(i) a parent of $u$ is removed from or inserted into the DAG when updating the RLSLP $\mathcal{G}^{R}$ 
(i.e., $u$ is a child of a node in $\mathcal{I}_{\expl} \cup \mathcal{I}_{\impl} \cup \mathcal{R}_{\expl} \cup \mathcal{R}_{\impl}$); and 
(ii) $u$ itself is neither inserted into nor removed from the DAG when updating the RLSLP $\mathcal{G}^{R}$. 
The set $\mathcal{U}_{\children}$ can be computed in 
$\mathcal{O}(|\mathcal{R}_{\expl}| + |\mathcal{R}_{\impl}| + |\mathcal{I}_{\expl}| + |\mathcal{I}_{\impl}|)$ time 
by modifying the algorithm for computing $\mathcal{R}_{\expl} \cup \mathcal{R}_{\impl}$. 
Here, $|\mathcal{R}_{\expl}| + |\mathcal{R}_{\impl}| + |\mathcal{I}_{\expl}| + |\mathcal{I}_{\impl}| = \mathcal{O}(W_{\max})$. 

\subsubsection{Computation of \texorpdfstring{$\mathcal{U}_{\change}$}{Uchange}}
%\subsubsection{Computation of $\mathcal{U}_{\change}$.}
Let $u \in \mathcal{U}_{\expl}$ be an explicit node that is neither inserted into nor removed from the DAG when updating the RLSLP $\mathcal{G}^{R}$. 
Then, $u \in \mathcal{U}_{\change}$ if and only if either $u \in \mathcal{U}_{\children}$ or 
$u$ is a descendant of a node in $\mathcal{U}_{\children}$ within distance $\lceil \alpha + \log B \rceil - 1$. 
Hence, analogously to the algorithm for $\mathcal{R}_{\expl} \cup \mathcal{R}_{\impl}$,
$\mathcal{U}_{\change}$ is computed by traversing the DAG from the nodes in $\mathcal{U}_{\children}$.
This traversal takes $\mathcal{O}(|\mathcal{U}_{\change}|)$ time,
where $|\mathcal{U}_{\change}| \leq W_{\max}$.

\subsubsection{Recomputation of the first \texorpdfstring{$\lceil \beta B / \log \sigma \rceil$}{beta B / log sigma} characters of \texorpdfstring{$L_i$}{Li} and \texorpdfstring{$R_i$}{Ri}}
%\subparagraph{Recomputation of the first $\lceil \beta B / \log \sigma \rceil$ characters in $L_i$ and $R_i$}
%For each node $u_i$, the first $\lceil \beta B / \log \sigma \rceil$ characters of $L_i$ and of $R_i$ are each obtained by a random access query in $\mathcal{O}(H' + \beta B)$ time.
For each $u_i \in \mathcal{U}_{\change} \cup \mathcal{I}_{\expl}$, we recompute
the prefixes of $L_i$ and $R_i$ of lengths
$\min\{|L_i|,\lceil \beta B / \log \sigma \rceil\}$ and
$\min\{|R_i|,\lceil \beta B / \log \sigma \rceil\}$, respectively, by random
access queries in $\mathcal{O}(H' + \beta B)$ time.

\subsubsection{Recomputation of 
\texorpdfstring{$\mathbb{P}_{\vOcc}(u_i)$}{PvOcc(ui)}.}
For each $u_i \in \mathcal{U}_{\change} \cup \mathcal{I}_{\expl}$, we also
recompute $\mathbb{P}_{\vOcc}(u_i)$ as follows.

We start with a path $\mathbb{Q}$ consisting only of $u_i$ and extend it to
$\mathbb{P}_{\vOcc}(u_i)$ by a recursive algorithm. 
The algorithm traverses the DAG from $u_i$ toward the root. 
Suppose that $\mathbb{Q}$ has already been extended to the last node $u_{i_1}$
of a path
$\mathbb{P} = u_{i_d} \rightarrow u_{i_{d-1}} \rightarrow \cdots \rightarrow u_{i_1}$
in $\pathset$.
Let $\eta$ be the number of keys stored in the per-path hash table for
$\mathcal{E}_{\expl}(\mathbb{P})$.
For each $x \in \{1,2,\ldots,\eta\}$, let
$e_x \in \mathcal{E}_{\expl}(\mathbb{P})$ be an inter-path edge corresponding
to the $x$-th key in the doubly linked list for these $\eta$ keys, and let
$z_x$ be the number of directed edges from $\mathrm{src}(e_x)$ to
$\mathrm{dst}(e_x)$.
The key list is ordered so that the destinations of the corresponding edges
are encountered in order when traversing $\mathbb{P}$ from $u_{i_1}$ toward
$u_{i_d}$.
Let $L = \lceil \alpha + \log B \rceil$. 
We distinguish the following three cases.

\subparagraph{Case 1: $\eta = 0$.}
No inter-path edge enters $\mathbb{P}$. 
Hence, the first node $u_{i_d}$ of $\mathbb{P}$ is the root of the DAG. 
Since the root is explicit, $\mathbb{P}$ consists only of this root. 
Therefore, the current path $\mathbb{Q}$ has reached the root, and we return
$\mathbb{Q}$ as $\mathbb{P}_{\vOcc}(u_i)$.

\subparagraph{Case 2: $\eta = 1$.}
The unique inter-path edge entering $\mathbb{P}$ must enter $u_{i_d}$;
otherwise, $u_{i_d}$ would be the root of the DAG, contradicting $\eta=1$.
Thus, $\mathrm{dst}(e_1)=u_{i_d}$.

We compute $z_1$ in $\mathcal{O}(1)$ time using
Lemma~\ref{lem:basic_queries}(ii) and the corresponding key-value pair
$(\mathrm{src}(e_1), \mathrm{dst}(e_1))$.
The algorithm then extends $\mathbb{Q}$ further along $\mathbb{P}$ whenever
possible.

If $|\mathbb{Q}| + d - 1 \geq L$, then the length bound is reached before
$\mathbb{Q}$ can be extended to $\mathrm{src}(e_1)$. 
Let $s = L - |\mathbb{Q}|$.
We extend $\mathbb{Q}$ to $u_{i_{s+1}}$ and return it as
$\mathbb{P}_{\vOcc}(u_i)$.

Assume now that $|\mathbb{Q}| + d - 1 < L$.
If $z_1 = 1$, then $\mathbb{Q}$ can be extended to $\mathrm{src}(e_1)$.
We extend $\mathbb{Q}$ to $\mathrm{src}(e_1)$ and recursively apply the
algorithm to the path in $\pathset$ that contains $\mathrm{src}(e_1)$.

Otherwise, $z_1 \neq 1$, and $u_{i_d}$ has at least two incoming edges in
the DAG. Hence, $\mathbb{Q}$ cannot be extended beyond $u_{i_d}$ while
preserving the required property. 
Therefore, we extend $\mathbb{Q}$ to $u_{i_d}$ and return it as
$\mathbb{P}_{\vOcc}(u_i)$.

\subparagraph{Case 3: $\eta \geq 2$.}
Let $k \in \{1,2,\ldots,d\}$ be the integer satisfying
$u_{i_k} = \mathrm{dst}(e_1)$.
Then $u_{i_k}$ has at least two incoming edges in the DAG.
Indeed, $u_{i_k}$ has the incoming edge $e_1$.
If $k < d$, then $u_{i_k}$ also has the intra-path incoming edge from
$u_{i_{k+1}}$.
Otherwise, $k=d$, and the ordering of the key list implies that at least two
inter-path edges enter $u_{i_d}$.

The algorithm extends $\mathbb{Q}$ further along $\mathbb{P}$ whenever
possible.
If $|\mathbb{Q}| + k - 1 < L$, then $\mathbb{Q}$ can be extended to
$u_{i_k}$ but cannot be extended beyond $u_{i_k}$.
Therefore, we extend $\mathbb{Q}$ to $u_{i_k}$ and return it as
$\mathbb{P}_{\vOcc}(u_i)$.

Otherwise, the length bound is reached before $\mathbb{Q}$ can be extended
beyond $u_{i_{s+1}}$, where $s = L - |\mathbb{Q}|$.
Therefore, we extend $\mathbb{Q}$ to $u_{i_{s+1}}$ and return it as
$\mathbb{P}_{\vOcc}(u_i)$.

This recursive algorithm makes $\mathcal{O}(H')$ recursive calls, each taking
$\mathcal{O}(1)$ time. Thus, $\mathbb{P}_{\vOcc}(u_i)$ is computed in
$\mathcal{O}(H')$ time.

\subsubsection{Running time of Phase 3.}
%\subparagraph{Running time of Phase 3.}
Computing $\mathcal{U}_{\change}$ takes $\mathcal{O}(W_{\max})$ time,
and updating the additional information takes $\mathcal{O}((|\mathcal{U}_{\change}| + |\mathcal{I}_{\expl}|)(H' + B))$ time.
Therefore, the total running time is $\mathcal{O}(W_{\max} + (|\mathcal{U}_{\change}| + |\mathcal{I}_{\expl}|)(H' + B))$.

%%%%%%%%%%%%%%%%%%%%%%%%%%%%%%%%%%%%%%%%%%%%%%%%%%%%%%%%%%%%%%%%%%%%%%
\subsection{Proof of Lemma~\ref{lem:ancestor_upper_bound}}\label{app:ancestor_upper_bound}
\subparagraph{Proof of $|\mathcal{T}_{\diff}| = \mathcal{O}((m' + H) H')$.}
In the derivation tree, 
the nodes in $\mathcal{T}_{\diff}$ form a tree $\mathcal{T}$ of height $H'$ with $\mathcal{O}(|\mathcal{T}_{\leaf}| + m')$ leaves.
Here, $\mathcal{T}_{\leaf}$ is the set of internal nodes in the derivation tree of $\mathcal{G}^{R'}$
such that each internal node is a leaf of $\mathcal{T}$.
The nodes in $\mathcal{T}_{\leaf}$ correspond to distinct runs in $Q_{L}$ and $Q_{R}$,
resulting in $|\mathcal{T}_{\leaf}| = \mathcal{O}(\rho_{L} + \rho_{R})$,
where $\rho_{L}$ (respectively, $\rho_{R}$) is the number of runs in $Q_{L}$ (respectively, $Q_{R}$).
This is because the nodes corresponding to each run have the same parent (see Appendix~\ref{app:popped_sequence}).
Since $\rho_{L}, \rho_{R} = \mathcal{O}(H)$,
the tree has $\mathcal{O}((m' + H) H')$ nodes, resulting in $|\mathcal{T}_{\diff}| = \mathcal{O}((m' + H) H')$. 

\subparagraph{Proof of $|\mathcal{I}_{\impl}| = \mathcal{O}((m' + H) H')$.}
The nodes in $\mathcal{I}_{\impl}$ correspond to distinct nodes in $\mathcal{T}_{\diff}$, 
resulting in $|\mathcal{I}_{\impl}| = \mathcal{O}((m' + H) H')$. 

\subparagraph{Proof of $|\mathcal{R}_{\impl}| = \mathcal{O}(H^2)$.}
The popped sequences $Q_{L}$ and $Q_{R}$ are embedded in the derivation tree of $\mathcal{G}^{R}$ at two positions $1$ and $s$ in $T$, respectively. 
Similar to $\mathcal{T}_{\diff}$, 
we consider the set $\mathcal{T}'_{\diff}$ obtained by collecting the ancestors of each root of the subtrees. 
Then, 
$\mathcal{T}'_{\diff}$ forms a tree $\mathcal{T}'$ of height $H$ with $\mathcal{O}(|\mathcal{T}'_{\leaf}|)$ leaves. 
Here, $\mathcal{T}'_{\leaf}$ is the set of internal nodes in the derivation tree of $\mathcal{G}^{R}$
such that each internal node is a leaf of $\mathcal{T}'$.
Similar to $\mathcal{T}_{\leaf}$, 
we obtain $|\mathcal{T}'_{\leaf}| = \mathcal{O}(H)$. 
The tree has $\mathcal{O}(H^{2})$ nodes, resulting in $|\mathcal{T}'_{\diff}| = \mathcal{O}(H^{2})$. 
The nodes in $\mathcal{R}_{\impl}$ correspond to distinct nodes in $\mathcal{T}'_{\diff}$, 
resulting in $|\mathcal{R}_{\impl}| = \mathcal{O}(H^{2})$. 

\subparagraph{Proof of $|\mathcal{I}_{\expl}| = \mathcal{O}(m' + H)$ and $|\mathcal{R}_{\expl}| = \mathcal{O}(H)$.}
The nodes in $\mathcal{I}_{\expl}$ correspond to the leaves and internal nodes of $\mathcal{T}$ with degree at least $2$.
The number of such nodes is $\mathcal{O}(m' + H)$, 
resulting in $|\mathcal{I}_{\expl}| = \mathcal{O}(m' + H)$. 
Similarly, we obtain $|\mathcal{R}_{\expl}| = \mathcal{O}(H)$. 

\subparagraph{Proof of $|\mathcal{U}_{\change}| = \mathcal{O}(HB)$.}
For each node $v \in \mathcal{T}_{\leaf} \cup \mathcal{T}'_{\leaf}$, 
let $f(v)$ be the lowest ancestor of $v$ such that 
$f'(v)$ is inserted into or removed from the DAG when the RLSLP is updated, 
where $f'(v)$ is the DAG-node corresponding to $f(v)$. 
Let $\tau = |\{ f'(v) \mid v \in \mathcal{T}_{\leaf} \cup \mathcal{T}'_{\leaf}\}|$.
Since each node in the DAG has at most two children, 
any node has at most $2^{\lceil \alpha + \log B \rceil}$ descendants within distance $\lceil \alpha + \log B \rceil$. 
Therefore, $|\mathcal{U}_{\change}| \leq 2^{\lceil \alpha + \log B \rceil}\tau$,
because each node in $\mathcal{U}_{\change}$ is a descendant of some node in
$\{ f'(v) \mid v \in \mathcal{T}_{\leaf} \cup \mathcal{T}'_{\leaf} \}$
at distance at most $\lceil \alpha + \log B \rceil$. 
Moreover, $\tau = \mathcal{O}(H)$ follows from
$|\mathcal{T}_{\leaf}|, |\mathcal{T}'_{\leaf}| = \mathcal{O}(H)$.
Hence, $|\mathcal{U}_{\change}| = \mathcal{O}(2^{\lceil \alpha + \log B \rceil} H) = \mathcal{O}(HB)$.

%%%%%%%%%%%%%%%%%%%%%%%%%%%%%%%%%%%%%%%%%%%%%%%%%%%%%%%%%%%%%%%%%%%%%%
\section{Details of the Deletion Operation}\label{app:deletion_operation}
We focus on the case $s \neq 1$ and $s+m'-1 \neq n$, so that
$T'$ is the concatenation of $T[1..s-1]$ and $T[s+m'..n]$.
Let $Q_{L}$ and $Q_{R}$ be the popped sequences of the prefix $T[1..s-1]$ and suffix $T[s+m'..n]$;
they are embedded as subtrees in the derivation tree of $\mathcal{G}^{R'}$ at positions $1$ and $s$ in $T'$, respectively,
and the computation of these subtrees can be skipped.

The update of the dynamic RR-index proceeds in three phases, paralleling the insertion algorithm.
Phase~1 constructs the derivation tree of $\mathcal{G}^{R'} = (\mathcal{V}', \Sigma', \mathcal{D}', E')$ by restricted recompression, computing the sequences $S'^{0}, S'^{1}, \ldots, S'^{H'}$ in a bottom-up manner.
During this construction, the nodes corresponding to the embeddings of $Q_{L}$ and $Q_{R}$ are skipped.
For each nonterminal in $\mathcal{V}' \setminus \mathcal{V}$, the corresponding node is inserted into the DAG, and the dynamic data structures are updated.
Phases~2 and~3 are identical to those of the insertion algorithm.

Thus, the expected amortized running time is $\mathcal{O}(m' \log^{2} n + \log^{3} n)$.

%% file: tex/9d_appendix.tex
\clearpage
\section{Details of Experiments}\label{app:experiment}
\begin{table}[t]
\scriptsize
\centering
\caption{
Per-dataset statistics, summarized in Section~\ref{sec:experiment} for \Dataset{chr19}. 
$\sigma$ is the alphabet size of the string $T$;
$n$ is the length of $T$; 
$\delta$ is the substring complexity of $T$;
$M$ is the number of explicit nodes in the DAG representing the derivation tree built by restricted recompression; $H$ is the height of the derivation tree.
$\occ_{10}$, $\occ_{100}$, and $\occ_{1000}$ denote the average number of occurrences in $T$ of 1{,}000 patterns of lengths 10, 100, and 1000, respectively, used for the locate queries.
}
\label{tab:dataset-stats_full}
\begin{tabular}{l||r|r|r|r|r|r|r|r}
\hline
Data & $\sigma$ & $n$ [$10^3$] & $\delta$ [$10^3$] & $M$ [$10^3$] & $H$ & $\occ_{10}$ & $\occ_{100}$ & $\occ_{1000}$ \\
\hline
\hline
\Dataset{cere}             & 5   & 461,287    & 1,003 & 4,874  & 304  & 1,046 & 29 & 6 \\
\Dataset{coreutils}        & 236 & 205,282    & 636   & 3,828  & 292  & 7,308 & 154 & 16 \\
\Dataset{einstein.de.txt}  & 117 & 92,758     & 16    & 122    & 272    & 5,983 & 615 & 130 \\
\Dataset{einstein.en.txt}  & 139 & 467,627    & 42    & 320    & 300    & 21,097 & 2,547 & 604 \\
\Dataset{Escherichia Coli} & 15  & 112,690    & 1,338 & 4,897  & 282  & 214 & 7 & 2 \\
\Dataset{influenza}        & 15  & 154,809    & 282   & 2,576  & 286  & 2,993 & 307 & 14 \\
\Dataset{kernel}           & 160 & 257,962    & 406   & 2,069  & 302  & 2,857 & 53 & 26 \\
\Dataset{para}             & 5   & 429,266    & 1,369 & 6,333  & 298  & 916 & 20 & 10 \\
\Dataset{world leaders}    & 89  & 46,968     & 69    & 580    & 260    & 24,326 & 1,217 & 10 \\
\hline
\Dataset{enwiki}           & 206 & 37,227,587 & 7,306 & 52,806 & 366 & 525,289 & 1,672 & 427 \\
\Dataset{chr19}            & 4   & 59,125,169 & 2,715 & 27,336 & 372 & 881,774 & 940 & 504 \\
\hline
\end{tabular}
\end{table}

\subsection{Dataset Statistics}\label{app:dataset-stats_full}
Table~\ref{tab:dataset-stats_full} reports the relevant statistics for every dataset.

\subsection{Results for Other Datasets}\label{app:result_other_datasets}
\begin{table}[t]
\scriptsize
\centering
\caption{
Average time and standard deviation for substring insertions on every dataset. 
For the (static) r-index, each row shows the construction time, since it does not support updates. 
The fastest dynamic method in each row is shown in bold.
}
\label{tab:insertion_time_full}
\begin{tabular}{l||r|r|r|r}
\hline
 & \multicolumn{4}{c}{Insertion Time for $m' = 1$ (ms)} \\
Data & Dynamic RR-index & Dynamic SE-index & Dynamic r-index & r-index \\
\hline
\hline
\Dataset{cere}             & $\mathbf{24 \pm 7}$  & $36 \pm 6$  & $62 \pm 115$    & 378,650 \\
\Dataset{coreutils}        & $\mathbf{27 \pm 7}$  & $31 \pm 7$  & $2,051 \pm 4,641$ & 143,300 \\
\Dataset{einstein.de.txt}  & $\mathbf{13 \pm 2}$  & $16 \pm 2$  & $221 \pm 177$   & 41,880 \\
\Dataset{einstein.en.txt}  & $\mathbf{14 \pm 2}$  & $24 \pm 3$  & $497 \pm 436$   & 207,580 \\
\Dataset{Escherichia Coli} & $\mathbf{22 \pm 8}$  & $31 \pm 6$  & $158 \pm 565$   & 152,760 \\
\Dataset{influenza}        & $21 \pm 4$  & $33 \pm 6$  & $\mathbf{6 \pm 8}$       & 109,180 \\
\Dataset{kernel}           & $43 \pm 6$  & $\mathbf{30 \pm 6}$  & $2,365 \pm 3,702$ & 154,020 \\
\Dataset{para}             & $\mathbf{25 \pm 9}$  & $64 \pm 8$  & $29 \pm 39$     & 409,420 \\
\Dataset{world leaders}    & $\mathbf{15 \pm 3}$  & $20 \pm 3$  & $65 \pm 331$    & 27,660 \\
\hline
\Dataset{enwiki}           & $\mathbf{51 \pm 56}$ & $117 \pm 46$ & $883 \pm 829$  & 38,921,100 \\
\Dataset{chr19}            & $\mathbf{55 \pm 29}$ & $73 \pm 25$ & $824 \pm 2,317$  & 52,106,240 \\
\hline
\multicolumn{5}{c}{} \\
\hline
 & \multicolumn{4}{c}{Insertion Time for $m' = 10$ (ms)} \\
Data & Dynamic RR-index & Dynamic SE-index & Dynamic r-index & r-index \\
\hline
\hline
\Dataset{cere}             & $\mathbf{24 \pm 7}$  & $37 \pm 7$  & $60 \pm 118$    & 378,650 \\
\Dataset{coreutils}        & $\mathbf{27 \pm 8}$  & $30 \pm 6$  & $2,099 \pm 4,414$ & 143,300 \\
\Dataset{einstein.de.txt}  & $\mathbf{13 \pm 2}$  & $17 \pm 2$  & $219 \pm 172$   & 41,880 \\
\Dataset{einstein.en.txt}  & $\mathbf{15 \pm 3}$  & $24 \pm 3$  & $503 \pm 445$   & 207,580 \\
\Dataset{Escherichia Coli} & $\mathbf{21 \pm 7}$  & $35 \pm 8$  & $133 \pm 516$   & 152,760 \\
\Dataset{influenza}        & $21 \pm 4$  & $32 \pm 5$  & $\mathbf{7 \pm 8}$       & 109,180 \\
\Dataset{kernel}           & $44 \pm 6$  & $\mathbf{27 \pm 4}$  & $2,150 \pm 3,417$ & 154,020 \\
\Dataset{para}             & $\mathbf{28 \pm 10}$ & $66 \pm 10$ & $32 \pm 45$     & 409,420 \\
\Dataset{world leaders}    & $\mathbf{15 \pm 3}$  & $20 \pm 3$  & $87 \pm 426$    & 27,660 \\
\hline
\Dataset{enwiki}           & $\mathbf{55 \pm 58}$ & $111 \pm 44$ & $818 \pm 770$  & 38,921,100 \\
\Dataset{chr19}            & $\mathbf{55 \pm 29}$ & $74 \pm 25$ & $919 \pm 2,781$  & 52,106,240 \\
\hline
\multicolumn{5}{c}{} \\
\hline
 & \multicolumn{4}{c}{Insertion Time for $m' = 100$ (ms)} \\
Data & Dynamic RR-index & Dynamic SE-index & Dynamic r-index & r-index \\
\hline
\hline
\Dataset{cere}             & $\mathbf{24 \pm 7}$  & $39 \pm 8$  & $56 \pm 114$    & 378,650 \\
\Dataset{coreutils}        & $\mathbf{29 \pm 8}$  & $34 \pm 7$  & $2,053 \pm 4,573$ & 143,300 \\
\Dataset{einstein.de.txt}  & $\mathbf{14 \pm 2}$  & $18 \pm 2$  & $226 \pm 180$   & 41,880 \\
\Dataset{einstein.en.txt}  & $\mathbf{15 \pm 2}$  & $25 \pm 3$  & $521 \pm 448$   & 207,580 \\
\Dataset{Escherichia Coli} & $\mathbf{21 \pm 7}$  & $33 \pm 7$  & $131 \pm 484$   & 152,760 \\
\Dataset{influenza}        & $21 \pm 4$  & $37 \pm 7$  & $\mathbf{7 \pm 8}$       & 109,180 \\
\Dataset{kernel}           & $47 \pm 6$  & $\mathbf{28 \pm 4}$  & $2,156 \pm 3,491$ & 154,020 \\
\Dataset{para}             & $\mathbf{29 \pm 10}$ & $71 \pm 11$ & $31 \pm 36$     & 409,420 \\
\Dataset{world leaders}    & $\mathbf{14 \pm 2}$  & $21 \pm 3$  & $72 \pm 382$    & 27,660 \\
\hline
\Dataset{enwiki}           & $\mathbf{50 \pm 55}$ & $117 \pm 43$ & $841 \pm 849$  & 38,921,100 \\
\Dataset{chr19}            & $\mathbf{62 \pm 31}$ & $83 \pm 29$ & $925 \pm 2,713$  & 52,106,240 \\
\hline
\multicolumn{5}{c}{} \\
\hline
 & \multicolumn{4}{c}{Insertion Time for $m' = 1000$ (ms)} \\
Data & Dynamic RR-index & Dynamic SE-index & Dynamic r-index & r-index \\
\hline
\hline
\Dataset{cere}             & $\mathbf{26 \pm 7}$  & $44 \pm 9$  & $60 \pm 104$    & 378,650 \\
\Dataset{coreutils}        & $\mathbf{28 \pm 6}$  & $40 \pm 8$  & $2,025 \pm 4,282$ & 143,300 \\
\Dataset{einstein.de.txt}  & $\mathbf{15 \pm 2}$  & $19 \pm 2$  & $223 \pm 173$   & 41,880 \\
\Dataset{einstein.en.txt}  & $\mathbf{15 \pm 2}$  & $27 \pm 3$  & $520 \pm 498$   & 207,580 \\
\Dataset{Escherichia Coli} & $\mathbf{23 \pm 7}$  & $37 \pm 7$  & $148 \pm 482$   & 152,760 \\
\Dataset{influenza}        & $23 \pm 4$  & $36 \pm 5$  & $\mathbf{11 \pm 8}$      & 109,180 \\
\Dataset{kernel}           & $53 \pm 9$  & $\mathbf{33 \pm 5}$  & $1,997 \pm 3,061$ & 154,020 \\
\Dataset{para}             & $\mathbf{32 \pm 11}$ & $75 \pm 9$  & $35 \pm 37$     & 409,420 \\
\Dataset{world leaders}    & $\mathbf{16 \pm 3}$  & $22 \pm 3$  & $86 \pm 435$    & 27,660 \\
\hline
\Dataset{enwiki}           & $\mathbf{58 \pm 56}$ & $125 \pm 43$ & $885 \pm 824$  & 38,921,100 \\
\Dataset{chr19}            & $\mathbf{63 \pm 30}$ & $81 \pm 25$ & $1,122 \pm 3,625$ & 52,106,240 \\
\hline
\end{tabular}
\end{table}

\begin{table}[t]
\scriptsize
\centering
\caption{
Average time and standard deviation for substring deletions on every dataset. For the (static) r-index, each row shows the construction time, since it does not support updates. 
The fastest dynamic method in each row is shown in bold.
}
\label{tab:deletion_time_full}
\begin{tabular}{l||r|r|r|r}
\hline
 & \multicolumn{4}{c}{Deletion Time for $m' = 1$ (ms)} \\
Data & Dynamic RR-index & Dynamic SE-index & Dynamic r-index & r-index \\
\hline
\hline
\Dataset{cere}             & $\mathbf{21 \pm 2}$ & $31 \pm 3$  & $55 \pm 109$    & 378,650 \\
\Dataset{coreutils}        & $\mathbf{22 \pm 3}$ & $26 \pm 4$  & $1,891 \pm 4,448$ & 143,300 \\
\Dataset{einstein.de.txt}  & $\mathbf{11 \pm 1}$ & $16 \pm 2$  & $215 \pm 183$   & 41,880 \\
\Dataset{einstein.en.txt}  & $\mathbf{13 \pm 1}$ & $23 \pm 2$  & $484 \pm 459$   & 207,580 \\
\Dataset{Escherichia Coli} & $\mathbf{19 \pm 3}$ & $27 \pm 3$  & $143 \pm 533$   & 152,760 \\
\Dataset{influenza}        & $18 \pm 2$ & $29 \pm 3$  & $\mathbf{6 \pm 7}$       & 109,180 \\
\Dataset{kernel}           & $32 \pm 3$ & $\mathbf{25 \pm 4}$  & $2,186 \pm 3,591$ & 154,020 \\
\Dataset{para}             & $\mathbf{22 \pm 3}$ & $53 \pm 5$  & $25 \pm 36$     & 409,420 \\
\Dataset{world leaders}    & $\mathbf{14 \pm 2}$ & $19 \pm 2$  & $60 \pm 313$    & 27,660 \\
\hline
\Dataset{enwiki}           & $\mathbf{42 \pm 8}$ & $92 \pm 13$ & $808 \pm 792$   & 38,921,100 \\
\Dataset{chr19}            & $\mathbf{49 \pm 9}$ & $62 \pm 6$  & $775 \pm 2,240$  & 52,106,240 \\
\hline
\multicolumn{5}{c}{} \\
\hline
 & \multicolumn{4}{c}{Deletion Time for $m' = 10$ (ms)} \\
Data & Dynamic RR-index & Dynamic SE-index & Dynamic r-index & r-index \\
\hline
\hline
\Dataset{cere}             & $\mathbf{21 \pm 2}$ & $31 \pm 3$ & $54 \pm 112$    & 378,650 \\
\Dataset{coreutils}        & $\mathbf{22 \pm 3}$ & $26 \pm 3$ & $1,921 \pm 4,200$ & 143,300 \\
\Dataset{einstein.de.txt}  & $\mathbf{12 \pm 2}$ & $16 \pm 2$ & $214 \pm 179$   & 41,880 \\
\Dataset{einstein.en.txt}  & $\mathbf{14 \pm 3}$ & $23 \pm 2$ & $488 \pm 450$   & 207,580 \\
\Dataset{Escherichia Coli} & $\mathbf{18 \pm 2}$ & $29 \pm 4$ & $119 \pm 494$   & 152,760 \\
\Dataset{influenza}        & $18 \pm 2$ & $28 \pm 3$ & $\mathbf{6 \pm 8}$       & 109,180 \\
\Dataset{kernel}           & $33 \pm 3$ & $\mathbf{24 \pm 2}$ & $1,978 \pm 3,263$ & 154,020 \\
\Dataset{para}             & $\mathbf{24 \pm 5}$ & $55 \pm 6$ & $28 \pm 43$     & 409,420 \\
\Dataset{world leaders}    & $\mathbf{14 \pm 2}$ & $19 \pm 2$ & $82 \pm 415$    & 27,660 \\
\hline
\Dataset{enwiki}           & $\mathbf{44 \pm 8}$ & $89 \pm 8$ & $744 \pm 736$   & 38,921,100 \\
\Dataset{chr19}            & $\mathbf{49 \pm 9}$ & $63 \pm 6$ & $855 \pm 2,636$  & 52,106,240 \\
\hline
\multicolumn{5}{c}{} \\
\hline
 & \multicolumn{4}{c}{Deletion Time for $m' = 100$ (ms)} \\
Data & Dynamic RR-index & Dynamic SE-index & Dynamic r-index & r-index \\
\hline
\hline
\Dataset{cere}             & $\mathbf{22 \pm 2}$  & $33 \pm 4$  & $51 \pm 108$    & 378,650 \\
\Dataset{coreutils}        & $\mathbf{23 \pm 4}$  & $28 \pm 4$  & $1,901 \pm 4,444$ & 143,300 \\
\Dataset{einstein.de.txt}  & $\mathbf{12 \pm 2}$  & $17 \pm 2$  & $224 \pm 195$   & 41,880 \\
\Dataset{einstein.en.txt}  & $\mathbf{14 \pm 1}$  & $24 \pm 2$  & $505 \pm 451$   & 207,580 \\
\Dataset{Escherichia Coli} & $\mathbf{19 \pm 2}$  & $29 \pm 3$  & $119 \pm 455$   & 152,760 \\
\Dataset{influenza}        & $19 \pm 2$  & $31 \pm 4$  & $\mathbf{6 \pm 7}$       & 109,180 \\
\Dataset{kernel}           & $34 \pm 3$  & $\mathbf{25 \pm 3}$  & $1,993 \pm 3,381$ & 154,020 \\
\Dataset{para}             & $\mathbf{25 \pm 6}$  & $57 \pm 6$  & $27 \pm 34$     & 409,420 \\
\Dataset{world leaders}    & $\mathbf{14 \pm 2}$  & $19 \pm 2$  & $67 \pm 369$    & 27,660 \\
\hline
\Dataset{enwiki}           & $\mathbf{42 \pm 5}$  & $92 \pm 8$  & $763 \pm 809$   & 38,921,100 \\
\Dataset{chr19}            & $\mathbf{54 \pm 12}$ & $69 \pm 11$ & $866 \pm 2,597$  & 52,106,240 \\
\hline
\multicolumn{5}{c}{} \\
\hline
 & \multicolumn{4}{c}{Deletion Time for $m' = 1000$ (ms)} \\
Data & Dynamic RR-index & Dynamic SE-index & Dynamic r-index & r-index \\
\hline
\hline
\Dataset{cere}             & $\mathbf{22 \pm 3}$  & $37 \pm 5$  & $56 \pm 99$     & 378,650 \\
\Dataset{coreutils}        & $\mathbf{23 \pm 2}$  & $32 \pm 5$  & $1,871 \pm 4,131$ & 143,300 \\
\Dataset{einstein.de.txt}  & $\mathbf{12 \pm 1}$  & $18 \pm 2$  & $218 \pm 178$   & 41,880 \\
\Dataset{einstein.en.txt}  & $\mathbf{14 \pm 1}$  & $25 \pm 3$  & $506 \pm 508$   & 207,580 \\
\Dataset{Escherichia Coli} & $\mathbf{20 \pm 2}$  & $31 \pm 3$  & $138 \pm 458$   & 152,760 \\
\Dataset{influenza}        & $20 \pm 2$  & $31 \pm 3$  & $\mathbf{13 \pm 8}$      & 109,180 \\
\Dataset{kernel}           & $36 \pm 5$  & $\mathbf{28 \pm 4}$  & $1,847 \pm 2,956$ & 154,020 \\
\Dataset{para}             & $\mathbf{27 \pm 6}$  & $60 \pm 6$  & $34 \pm 35$     & 409,420 \\
\Dataset{world leaders}    & $\mathbf{15 \pm 2}$  & $20 \pm 2$  & $84 \pm 429$    & 27,660 \\
\hline
\Dataset{enwiki}           & $\mathbf{47 \pm 8}$  & $97 \pm 9$  & $810 \pm 783$   & 38,921,100 \\
\Dataset{chr19}            & $\mathbf{55 \pm 11}$ & $69 \pm 7$  & $1,038 \pm 3,357$ & 52,106,240 \\
\hline
\end{tabular}
\end{table}

\begin{table}[t]
\scriptsize 
\centering
\caption{Average locate time and standard deviation for patterns 
of length $m \in \{10, 100, 1000\}$ on every dataset.
The fastest dynamic method in each row is shown in bold.
}
\label{tab:locate-time_full}
\begin{tabular}{l||r|r|r|r}
\hline
 & \multicolumn{4}{c}{Locate Time for $m = 10$ (ms)} \\
Data & Dynamic RR-index & Dynamic SE-index & Dynamic r-index & r-index \\
\hline\hline
\Dataset{cere}             & $\mathbf{1.8 \pm 2.9}$ & $3.5 \pm 6.5$ & $3.2 \pm 6.9$ & $0.4 \pm 0.9$ \\
\Dataset{coreutils}        & $\mathbf{8.3 \pm 51.1}$ & $13.6 \pm 71.5$ & $11.4 \pm 48.8$ & $1.5 \pm 6.0$ \\
\Dataset{einstein.de.txt}  & $\mathbf{1.1 \pm 2.9}$ & $2.5 \pm 9.7$ & $6.6 \pm 26.5$ & $1.0 \pm 4.0$ \\
\Dataset{einstein.en.txt}  & $\mathbf{2.3 \pm 6.2}$ & $7.7 \pm 27.5$ & $22.0 \pm 81.0$ & $2.5 \pm 9.0$ \\
\Dataset{Escherichia Coli} & $\mathbf{1.1 \pm 0.8}$ & $1.9 \pm 1.1$ & $1.3 \pm 1.3$ & $0.2 \pm 0.1$ \\
\Dataset{influenza}        & $\mathbf{4.2 \pm 4.5}$ & $7.7 \pm 6.3$ & $5.2 \pm 4.3$ & $0.7 \pm 0.6$ \\
\Dataset{kernel}           & $\mathbf{1.3 \pm 5.3}$ & $2.8 \pm 11.1$ & $4.7 \pm 21.4$ & $0.7 \pm 3.2$ \\
\Dataset{para}             & $\mathbf{2.0 \pm 4.4}$ & $3.7 \pm 4.2$ & $3.2 \pm 3.7$ & $0.4 \pm 0.5$ \\
\Dataset{world leaders}    & $27.4 \pm 78.5$ & $\mathbf{25.4 \pm 61.4}$ & $31.1 \pm 84.4$ & $3.2 \pm 8.9$ \\
\hline
\Dataset{enwiki} & $\mathbf{101.8 \pm 483.7}$ & $532.1 \pm 2{,}979.3$ & $1{,}093.3 \pm 6{,}018.5$ & $96.3 \pm 546.9$ \\
\Dataset{chr19} & $\mathbf{330.0 \pm 781.4}$ & $2{,}362.9 \pm 6{,}722.4$ & $3{,}793.2 \pm 10{,}835.2$ & $298.1 \pm 845.1$ \\
\hline 
\multicolumn{5}{c}{} \\
\hline
 & \multicolumn{4}{c}{Locate Time for $m = 100$ (ms)} \\
Data & Dynamic RR-index & Dynamic SE-index & Dynamic r-index & r-index \\
\hline\hline
\Dataset{cere}             & $\mathbf{0.6 \pm 2.1}$ & $1.7 \pm 1.5$ & $\mathbf{0.6 \pm 0.6}$ & $0.2 \pm 0.1$ \\
\Dataset{coreutils}        & $\mathbf{1.0 \pm 1.5}$ & $2.0 \pm 1.6$ & $\mathbf{1.0 \pm 0.9}$ & $0.3 \pm 0.2$ \\
\Dataset{einstein.de.txt}  & $\mathbf{0.7 \pm 0.3}$ & $1.4 \pm 0.4$ & $1.2 \pm 0.7$ & $0.3 \pm 0.1$ \\
\Dataset{einstein.en.txt}  & $\mathbf{1.3 \pm 0.9}$ & $2.6 \pm 1.2$ & $3.4 \pm 2.3$ & $0.6 \pm 0.3$ \\
\Dataset{Escherichia Coli} & $\mathbf{0.4 \pm 0.1}$ & $1.5 \pm 0.2$ & $0.6 \pm 0.1$ & $0.1 \pm 0.0$ \\
\Dataset{influenza}        & $1.9 \pm 2.5$ & $3.4 \pm 3.3$ & $\mathbf{0.9 \pm 0.9}$ & $0.2 \pm 0.1$ \\
\Dataset{kernel}           & $\mathbf{0.6 \pm 0.2}$ & $1.4 \pm 0.3$ & $1.0 \pm 0.5$ & $0.2 \pm 0.1$ \\
\Dataset{para}             & $\mathbf{0.5 \pm 0.3}$ & $1.6 \pm 0.2$ & $0.6 \pm 0.1$ & $0.2 \pm 0.0$ \\
\Dataset{world leaders}    & $2.5 \pm 6.5$ & $3.5 \pm 6.6$ & $\mathbf{2.2 \pm 5.0}$ & $0.4 \pm 0.7$ \\
\hline
\Dataset{enwiki} & $\mathbf{2.1 \pm 1.5}$ & $7.3 \pm 4.9$ & $5.3 \pm 4.3$ & $0.8 \pm 0.5$ \\
\Dataset{chr19} & $\mathbf{1.3 \pm 0.5}$ & $8.3 \pm 2.9$ & $5.2 \pm 1.7$ & $0.6 \pm 0.2$ \\
\hline 
\multicolumn{5}{c}{} \\
\hline
 & \multicolumn{4}{c}{Locate Time for $m = 1000$ (ms)} \\
Data & Dynamic RR-index & Dynamic SE-index & Dynamic r-index & r-index \\
\hline\hline
\Dataset{cere}               & $\mathbf{1.5 \pm 0.3}$ & $4.2 \pm 0.7$ & $3.8 \pm 0.3$ & $1.1 \pm 0.2$ \\
\Dataset{coreutils}          & $\mathbf{2.1 \pm 0.4}$ & $4.1 \pm 0.5$ & $6.7 \pm 0.3$ & $1.5 \pm 0.4$ \\
\Dataset{einstein.de.txt}    & $\mathbf{1.6 \pm 0.5}$ & $3.2 \pm 0.9$ & $3.6 \pm 0.7$ & $1.3 \pm 0.4$ \\
\Dataset{einstein.en.txt}    & $\mathbf{2.1 \pm 0.8}$ & $3.9 \pm 0.9$ & $5.3 \pm 0.9$ & $1.1 \pm 0.2$ \\
\Dataset{Escherichia Coli}   & $\mathbf{1.5 \pm 0.2}$ & $3.8 \pm 0.3$ & $4.8 \pm 0.2$ & $1.0 \pm 0.2$ \\
\Dataset{influenza}          & $\mathbf{2.4 \pm 0.8}$ & $4.4 \pm 1.1$ & $3.9 \pm 0.2$ & $1.0 \pm 0.2$ \\
\Dataset{kernel}             & $\mathbf{1.8 \pm 0.3}$ & $3.8 \pm 0.6$ & $6.3 \pm 0.2$ & $1.4 \pm 0.4$ \\
\Dataset{para}               & $\mathbf{1.5 \pm 0.3}$ & $4.2 \pm 0.5$ & $4.0 \pm 0.3$ & $0.9 \pm 0.1$ \\
\Dataset{world leaders}      & $\mathbf{1.6 \pm 0.4}$ & $2.9 \pm 0.7$ & $4.1 \pm 0.5$ & $1.2 \pm 0.4$ \\
\hline
\Dataset{enwiki}             & $\mathbf{3.6 \pm 0.9}$ & $8.4 \pm 1.7$ & $14.1 \pm 3.4$ & $2.8 \pm 2.0$ \\
\Dataset{chr19}              & $\mathbf{3.2 \pm 1.0}$ & $9.6 \pm 1.9$ & $8.4 \pm 1.8$ & $2.4 \pm 2.0$ \\
\hline

\end{tabular}

\end{table}

\begin{table}[t]
\scriptsize 
\centering
\caption{Working space during locate queries for all datasets.
The smallest working space among the dynamic methods in each row is shown in bold.}
\label{tab:index_size_full}
\begin{tabular}{l||r|r|r|r}
\hline
 & \multicolumn{4}{c}{Working Space (MiB)} \\
Data & Dynamic RR-index & Dynamic SE-index & Dynamic r-index & r-index \\
\hline
\hline
\Dataset{cere}              & 685 & 826 & \textbf{260} & 110 \\
\Dataset{coreutils}         & 522 & 602 & \textbf{216} & 48 \\
\Dataset{einstein.de.txt}   & 19 & 22 & \textbf{12} & 2 \\
\Dataset{einstein.en.txt}   & 53 & 54 & \textbf{38} & 4 \\
\Dataset{Escherichia Coli}  & 691 & 781 & \textbf{448} & 127 \\
\Dataset{influenza}         & 354 & 420 & \textbf{77} & 29 \\
\Dataset{kernel}            & 283 & 324 & \textbf{137} & 29 \\
\Dataset{para}              & 886 & 1{,}040 & \textbf{473} & 145 \\
\Dataset{world leaders}     & 84 & 102 & \textbf{42} & 6 \\
\hline
\Dataset{enwiki} & 7{,}515 & 8{,}332 & \textbf{1{,}482} & 887 \\
\Dataset{chr19} & 3{,}787 & 4{,}810 & \textbf{2{,}297} & 541 \\
\hline
\end{tabular}
\end{table}

\begin{table}[t]
\scriptsize 
\centering
\caption{Construction time for all datasets.
The fastest construction time among the dynamic methods in each row is shown in bold.}
\label{tab:construction-time_full}
\begin{tabular}{l||r|r|r|r}
\hline
 & \multicolumn{4}{c}{Construction Time (min)} \\
Data & Dynamic RR-index & Dynamic SE-index & Dynamic r-index & r-index \\
\hline\hline
\Dataset{cere}             & \textbf{4.9} & 5.3 & 6.3 & 5.7 \\
\Dataset{coreutils}        & 2.8 & 3.2 & \textbf{2.4} & 2.2 \\
\Dataset{einstein.de.txt}  & 0.9 & \textbf{0.4} & 0.7 & 0.6 \\
\Dataset{einstein.en.txt}  & 4.6 & \textbf{2.1} & 3.5 & 3.1 \\
\Dataset{Escherichia Coli} & \textbf{1.9} & 3.5 & 2.5 & 1.9 \\
\Dataset{influenza}        & 1.9 & 2.2 & \textbf{1.8} & 1.6 \\
\Dataset{kernel}           & 2.9 & \textbf{2.4} & 2.6 & 2.5 \\
\Dataset{para}             & \textbf{5.1} & 6.2 & 6.8 & 5.9 \\
\Dataset{world leaders}    & \textbf{0.4} & \textbf{0.4} & 0.5 & 0.4 \\
\hline
\Dataset{enwiki}           & 469.4 & \textbf{315.1} & 648.7 & 664.1 \\
\Dataset{chr19}            & 554.6 & \textbf{285.8} & 868.4 & 863.4 \\
\hline
\end{tabular}
\end{table}
\subparagraph{Results for update operations.}
%Table~\ref{tab:insertion_time_full} reports the average time and standard deviation of substring insertions on all datasets. 
%See the full version~\cite{DBLP:journals/corr/abs-2604-24080} for
%the corresponding results for substring deletions. 
Tables~\ref{tab:insertion_time_full}--\ref{tab:deletion_time_full} report the average time and standard deviation for string insertions and substring deletions on all datasets. 
For update operations, the dynamic RR-index was either the fastest dynamic index or competitive with the fastest one in most settings. 
On \Dataset{enwiki} with $m' = 1000$, it was approximately $2.1\times$ faster than the dynamic SE-index and $16\times$ faster than the dynamic r-index. On the Pizza\&Chili repetitive corpus with $m' = 1000$, the best-case speedups were approximately $2.3\times$ over the dynamic SE-index (\Dataset{para}) and $77\times$ over the dynamic r-index (\Dataset{coreutils}), with insertions and deletions considered together.

\subparagraph{Results for locate queries.}
Table~\ref{tab:locate-time_full} reports the average time and standard deviation for 1{,}000 locate queries with short, medium, and long patterns on all datasets.
On \Dataset{enwiki}, the dynamic RR-index matched the r-index for short patterns and was only about $2.6\times$ slower for medium patterns; against the other dynamic indexes, it was at least $5\times$ faster for short patterns.

On most datasets and pattern lengths, the r-index was the fastest, and the dynamic
RR-index was usually the fastest among the dynamic indexes, with a few exceptions:
the dynamic SE-index was faster on \Dataset{world leaders} for short patterns, and
the dynamic r-index was faster on \Dataset{influenza} and \Dataset{world leaders}
for medium patterns.
For long patterns, however, it was at most about $2.5\times$ slower than the r-index. In the best case, the dynamic RR-index was at least $3\times$ faster than both other dynamic indexes for short patterns, at least $2\times$ for medium patterns, and at least $2.6\times$ for long patterns.

\subparagraph{Results for memory consumption.}
Table~\ref{tab:index_size_full} reports the working space used during locate queries on all datasets. 
Among the three dynamic indexes, the dynamic r-index was the most space-efficient, followed by the dynamic RR-index. On \Dataset{enwiki}, the dynamic RR-index used approximately $5\times$ the memory of the dynamic r-index. On the Pizza\&Chili repetitive corpus, this ratio ranged from $1.4\times$ (\Dataset{einstein.en.txt}) to $4.6\times$ (\Dataset{influenza}).

\subparagraph{Results for construction time.}
Table~\ref{tab:construction-time_full} reports the construction times of the dynamic RR-index and the other three methods on all datasets. On \Dataset{enwiki}, all four methods finished within about 11 hours; the dynamic SE-index was the fastest (about 5 hours), followed by the dynamic RR-index (about 8 hours). On the Pizza\&Chili repetitive corpus, all methods finished within 7 minutes for every dataset, with no significant difference among the four.